\documentclass[12pt]{article}
\usepackage{amsfonts,amsmath}
\usepackage[mathscr]{eucal}
\usepackage{amssymb}
\usepackage{amsthm}
\theoremstyle{plain}
\newtheorem{thm}{Theorem}

\def\theequation{\arabic{section}.\arabic{equation}}
\newcommand{\be}{\begin{eqnarray}}
\newcommand{\ee}{\end{eqnarray}}
\newcommand{\nn}{\nonumber \\}
\newcommand{\lb}{\label}
\newcommand{\p}[1]{(\ref{#1})}

\begin{document}

\begin{titlepage}

\vspace*{0.2cm}

\renewcommand{\thefootnote}{\star}
\begin{center}

{\LARGE\bf  ${\cal N}{=}\,4$ mechanics with diverse  ${\bf (4, 4, 0)}$ multiplets: }\\

\vspace{0.5cm}

{\LARGE\bf  explicit examples of }\\

\vspace{0.5cm}

{\LARGE\bf  HKT, CKT, and OKT geometries }

\vspace{1.5cm}
\renewcommand{\thefootnote}{$\star$}

{\large\bf Sergey~Fedoruk},\footnote{\,\,On leave of absence
from V.N.\,Karazin Kharkov National University, Ukraine}
\quad {\large\bf Evgeny~Ivanov}
 \vspace{0.5cm}

{\it Bogoliubov Laboratory of Theoretical Physics, JINR,}\\
{\it 141980 Dubna, Moscow region, Russia} \\
\vspace{0.1cm}

{\tt fedoruk,eivanov@theor.jinr.ru}\\
\vspace{0.7cm}

{\large\bf Andrei~Smilga} \vspace{0.5cm}

{\it SUBATECH, Universit\'e de Nantes,}\\
{\it 4 rue Alfred Kastler, BP 20722, Nantes 44307, France;}\\
{\it On leave of absence from ITEP, Moscow, Russia} \\
\vspace{0.1cm}

{\tt smilga@subatech.in2p3.fr}\\

\end{center}
\vspace{0.2cm} \vskip 0.6truecm \nopagebreak

\begin{abstract}
\noindent
We present simple models of ${\cal N}{=}\,4$ supersymmetric mechanics with ordinary and mirror linear ({\bf 4, 4, 0})
multiplets
that give a transparent description of HKT, CKT, and OKT geometries. These models are treated in the ${\cal N}{=}\,4$ and ${\cal N}{=}2\,$ superfield approaches,
as well as in the component approach. Our study makes manifest that the
CKT and OKT supersymmetric sigma models are distinguished from the more simple HKT models by the presence of
extra holomorphic torsions in the supercharges.

\end{abstract}

\vspace{1cm}

\phantom{Keywords: }

\newpage

\end{titlepage}

\setcounter{footnote}{0}

\setcounter{equation}0
\section{Introduction}
\subsection{SQM models and geometry}
It is known since \cite{Witten} that supersymmetric quantum mechanical (SQM)  sigma models describing the motion of a supersymmetric
particle on a curved manifold can be mapped to certain geometrical structures referred to as {\it complexes}.
This nice  mathematical interpretation of these models, combined with their use in a more physical context, e.g., with the fact that they
give rise to the effective theories of black holes (see, for example, \cite{MStrom,MS2,Britto,GutP1,GutP2} and relevant references),
make them an interesting subject for the intensive study which lasts for more than 30 years.

The simplest such model involving $D$ real  ({\bf 1, 2, 1}) multiplets $X^\mu$
 \footnote{We follow the notation of \cite{PT} such that the numerals count the numbers of the physical bosonic,
physical fermionic and auxiliary bosonic fields.}
and the action
\be
\label{Rham}
S \ = \  -\frac 12 \int dtd\bar\theta d\theta  \, g_{\mu\nu}(X) DX^\mu
\bar D X^\nu
    \ee
(with $g_{\mu\nu} = g_{\nu\mu}$) corresponds to the de Rham complex. Another basic model involves $d$ complex chiral ({\bf 2, 2, 0}) multiplets $Z^j$ ($\bar D Z^j = 0$),
has the action \cite{Hull,IS}
\be
\lb{Dolb}
 S\  = -\frac{1}{4}  \int dt d^2\theta \,
h_{j\bar k}(Z, \bar Z)\, D Z{\,}^j \bar D\bar Z{\,}^{\bar k} \, ,
\ee
and
describes the Dolbeault complex. These models possess the simplest ${\cal N} = 2$ supersymmetry~\footnote{${\cal N}$
counts the number of real supercharges.}.

The action \p{Rham} can be generalized to include an extra
 potential \cite{Witten} and/or torsions \cite{Kimura,FIS}.
An analogous generalization of the action \p{Dolb} includes extra gauge fields
\cite{Hull,IS} and/or holomorphic torsions \cite{Hull,FIS}. Gauge
fields are described by the superfield Wess-Zumino term
 \be
\lb{WZ}
S_{\rm gauge} \ =\ \int \, dt d^2\theta \, W(Z^j, \bar Z^j) \, .
 \ee
 The Dolbeault models with holomorphic torsions
 are obtained when one adds to the action \p{Dolb} the terms
 \be
\lb{holtorsion}
\Delta S \sim \ \int \, dt d^2\theta \,  {\cal B}_{jk} D Z^j DZ^k \ +\ {\rm c.c.}
 \ee
with complex antisymmetric ${\cal B}_{jk}$. The holomorphic exterior derivative of the 2-form ${\cal B}$ gives
the holomorphic torsion 3-form.

In certain cases, the models \p{Rham}, \p{Dolb} and their generalizations enjoy extended supersymmetries. It is well known, e.g.,  that
the   ${\cal N} = 2$ supersymmetry of the real sigma
model \p{Rham} can be extended to  ${\cal N} = 4$ if the manifold is K\"ahler
\cite{Zumino} and to  ${\cal N} = 8$ if the manifold is hyper-K\"ahler \cite{A-GF}. The emerging mathematical constructions can be called
{\it K\"ahler - de Rham} complex and {\it hyper-K\"ahler - de Rham} complex.

Likewise, the supersymmetry of the complex model \p{Dolb} can be extended to  ${\cal N} = 4$
when the metric is hyper-K\"ahler or has the so called HKT form
\cite{HoweP,GPS,Hull,Wipf}.
It was also  noted in \cite{HoweP,GPS,Hull} that  complex sigma models admit an
 extended ${\cal N}{=}\,4$ supersymmetry for  a geometry more general than HKT.
This geometry with certain relaxed (compared to HKT) conditions for complex structures
has been termed CKT  in \cite{Hull}.
For some special CKT metrics (the so called OKT metrics), the
models enjoying extended ${\cal N} = 8$ supersymmetry can be written.

Before explaining (we will do it shortly) what exactly HKT, CKT and OKT mean,
let us say a few words about terminology.
The abbreviation HKT stands for the {\it Hyper-K\"ahler with Torsions} geometry.
It is worth mentioning that, generically, the HKT manifolds
are {\it not} hyper-K\"ahler and not even K\"ahler. Their characteristic feature is the presence of three complex
structures which form a quaternion algebra
and are covariantly constant with respect to the appropriate connection with torsion.
Likewise, CKT means {\it Clifford K\"ahler with Torsions}.
This geometry is characterized by three complex structures which form the Clifford algebra but,
in general, not the quaternion one.
Finally, OKT ({\it Octonionic K\"ahler with Torsions}) manifolds are the manifolds of a
dimension which is an integer multiple of 8, such that their geometry
involves seven
 different complex structures forming the $7D$
 Clifford algebra. These structures reveal some relation to the octonion algebra, though do not satisfy it.
Neither CKT nor OKT manifolds are K\"ahler. Despite this mismatch (and the similar one for the generic HKT manifolds), we will follow the established literature tradition and use
the names HKT, CKT and OKT  for the considered type of geometries.

The HKT geometry introduced in \cite{HoweP,GPS} is by now well understood by mathematicians
\cite{Verb}.
 This could not be really
said, however, about the CKT and OKT geometries. One of the aims of our paper is to treat in detail
some simple particular examples of the CKT and OKT manifolds to make  more transparent
 their mathematical structure.

  \subsection{Overview of definitions and motivations}

Before proceeding with explicit calculations, let us make some preliminary remarks elucidating the
geometric structures encoded  in the HKT, CKT, and OKT SQM models and explaining why these models are worthy
to study.

The mathematical definition of the HKT geometry is the following \cite{Verb}.

\vspace{1mm}

HKT manifold is a manifold of dimension $D = 4n$ involving three different
integrable
\footnote{{\it Integrable} means that complex coordinates  can be introduced such that the metric has a Hermitian
form $ds^2 = 2h_{j \bar k} dz^j d\bar z^{\bar k}$. This is possible when the
{\it Nijenhuis tensor} vanishes. The latter condition can be written as
 \be
\lb{integrability}
\partial_{[\mu} (I^a)_{\nu]}^{\ \lambda} \ =\ (I^a)_\mu^{\ \sigma}    (I^a)_\nu^{\ \rho } \partial_{[\sigma} (I^a)_{\rho]}^{\ \lambda} \, ,
 \ee
where one  can equivalently put the  Levi-Civita covariant derivatives in place of the usual ones.
When \p{integrability} holds, the holomorphic exterior derivatives associated with the complex structures $I^a$ are nilpotent and can be interpreted
as supercharges.}
  antisymmetric complex structures $(I^a)_{\mu\nu} = -(I^a)_{\nu\mu}$ which satisfy the quaternion
algebra
 \be
\lb{quaternion}
(I^a)_\mu^{\ \lambda} (I^b)_\lambda^{\ \nu}  = -\delta^{ab} \delta_\mu^\nu + \epsilon^{abc} (I^c)_\mu^{\ \nu} \, .
 \ee
The standard torsionless Levi-Civita covariant derivatives of the complex structures $\nabla_\lambda I^a$ do not necessarily vanish
(if they do, this is a hyper-K\"ahler manifold).
However, for any complex structure $I$, one can always define a class of torsionful affine connections with respect to which both the metric
and this complex structure
are covariantly constant, $\tilde {\nabla}_\lambda g_{\mu\nu} = \tilde{\nabla}_\lambda (I)_{\mu\nu} = 0$.
If we also require  the torsion tensor $C_{\mu\nu\lambda}$ to be totally antisymmetric, such connection
is unique and is called {\it Bismut connection} \cite{Bismut}. In what follows, we will denote the torsion
entering the Bismut connection as $B_{\mu\nu\lambda}$. The explicit expression of the Bismut torsion tensor
for the complex structure $I$ is \cite{Hull}
 \be
\label{Bismut}
B_{\mu\nu\lambda} \ =\ I_\mu^{\ \alpha}  I_\nu^{\ \beta}
I_\lambda^{\ \gamma} (\nabla_\alpha I_{\beta\gamma} +   \nabla_\beta I_{\gamma\alpha} + \nabla_\gamma I_{\alpha\beta} )=
-\left(I_\mu^{\ \alpha} \nabla_\alpha I_{\nu\lambda} +
I_\nu^{\ \alpha}\nabla_\alpha I_{\lambda\mu} + I_\lambda^{\ \alpha}\nabla_\alpha I_{\mu\nu} \right)\, ,
 \ee
where $\nabla_\alpha$ are the Levi-Civita covariant derivatives.
This representation for $B_{\mu\nu\lambda}$ can be directly derived
from the condition of the covariant constancy for $I$ and
$g_{\mu\nu}$ with the Bismut connection. Note that the (anti)holomorphic
projections of the tensor \p{Bismut} vanish,
 \be
\label{holproj}
   (P^\pm)_\mu^{\ \alpha} (P^\pm)_\nu^{\ \beta}
(P^\pm)_\lambda^{\ \gamma} \, B_{\alpha\beta\gamma} \ =\ 0
 \ee
with  $(P^\pm)_\mu^{\ \alpha} = (\delta_\mu^\alpha \pm i I_\mu^{\ \alpha})/2$. Equivalently, this property
can be formulated in terms of the tensor $H_{\mu\nu\lambda}(C,I)$ defined as
\be
\lb{holdobavka}
H_{\mu\nu\lambda}(C,I) \ :=\  \frac 14\left( C_{\mu\nu\lambda}
-   I_\mu^{\ \rho}  I_\nu^{\ \sigma}  C_{\rho \sigma \lambda} -
I_\lambda^{\ \rho}  I_\mu^{\ \sigma}  C_{\rho \sigma\nu} -
I_\nu^{\ \rho}  I_\lambda^{\ \sigma}  C_{\rho \sigma \mu} \right).
\ee
Then eqs. \p{holproj} amount to the condition
\be
\quad H_{\mu\nu\lambda}(B,I) \ =\ 0\,.
 \ee
The origin of the coefficient $1/4$ in the definition of $H_{\mu\nu\lambda}(C,I)$ will be explained in Sect. 4.5; it ensures the projector-like property
$ H_{\mu\nu\lambda}(H,I) = H_{\mu\nu\lambda}\,$.

{}For an HKT manifold, the Bismut connections for all three structures $I^a$ {\it coincide}
and the covariant constancy condition can be written as
\be
 \lb{HKT}
\tilde{\nabla}_\lambda I^a_{\,\mu\,}{}^\nu  = 0  \, .
 \ee
One can then define   derivatives holomorphic with respect to these complex structures  and twisted by the presence of some particular extra gauge fields.
These objects can be shown to satisfy the ${\cal N}=4, d=1$ superalgebra \cite{Verb}.

A CKT manifold is a manifold with three antisymmetric integrable complex structures,
which do not necessarily satisfy the quaternion algebra \p{quaternion}.
They are required, however,
to form the {\it Clifford} algebra
  \be
 \lb{Clifford}
I^a I^b + I^b I^a   = - 2\delta^{ab}  \, .
 \ee
Further, an affine connection with the totally antisymmetric torsion tensor is
required to exist, such that the complex
 structures satisfy the conditions which are weaker than \p{HKT}:
\begin{equation}
\label{w-cov-compl}
\tilde\nabla_{(\lambda} I^{{a}}_{\mu)}{}^{\nu}=0\,,
\end{equation}
where $(\,\,)$ means symmetrization of the indices $\lambda, \mu\,$.

In order that the supercharges associated with the complex structures $I^a$ form the supersymmetry algebra,
two additional requirements are needed. First, the {\it Nijenhuis concomitants}
 \be
\lb{concom}
{\textstyle\frac{1}{2}}\, N(a,b)_{\mu\nu}^\lambda \ =\ \left\{ (I^a)_{[\mu}^{\ \sigma} \partial_\sigma (I^b)_{\nu]}^{\ \lambda} -
\partial_{[\mu} (I^b)_{\nu]}^{\ \sigma}  (I^a)_\sigma^{\ \lambda}  \right\} \ + (a \leftrightarrow b)
 \ee
should vanish,
\be
\lb{concom1}
N(a,b)_{\mu\nu}^\lambda \ = 0\,.
 \ee

Second, the torsion tensor should satisfy the relation
\begin{equation}
\lb{IdC}
(I^a)^\tau{}_{[\mu} \partial_{|\tau|} C_{\nu\lambda\rho]}=
(I^a)^\tau{}_{[\mu} \partial_{\nu} C_{\lambda\rho]\tau}+2C_{\tau[\mu\nu} \partial_{\lambda}(I^a)^\tau{}_{\rho]} ,
\end{equation}
which can be concisely rewritten as \cite{GPS,Hull}
\begin{equation}\label{eq-cs1a}
\iota_a d C={\textstyle\frac{2}{3}}\,d (\iota_a C)\,.
\end{equation}
 The operator $\iota_a$ acts on a generic $n$-form $\omega$ according to the rule:
 \be
\lb{iotaHull}
&&{\rm if} \ \ \ \ \ \ \ \ \ \ \omega \ =\ \frac 1{n!} \omega_{\mu_1 \ldots \mu_n} dx^{\mu_1} \wedge \cdots \wedge dx^{\mu_n}\,, \nn
&&{\rm then} \ \ \ \ \iota_a \omega \ =\ \frac 1{(n-1)!} \omega_{\nu [\mu_2 \ldots \mu_{n-1}} (I^a)^{\nu}_{\ \mu_1]}\,.
 dx^{\mu_1} \wedge \cdots \wedge dx^{\mu_n}\,.
 \ee
For a form $\omega_{p,q}$ with $p$ holomorphic (with respect to the complex structure $I^a$) and $q$ antiholomorphic indices, the action
of $\iota_a$ is reduced to the multiplication by $i(p-q)$.

\vspace{1mm}

An OKT manifold is a manifold with {\it seven} integrable complex structures
which satisfy the Clifford algebra and all the additional conditions listed above.

This definition of the CKT and OKT manifolds, though being of general use, looks less transparent and geometrical  as compared, e.g., with
that of  the HKT manifolds which involves, as the main condition, the coincidence of the Bismut connections for all three relevant
complex structures. In Sect. 5 we will suggest another equivalent definition of the CKT and OKT manifolds based on a generalization
of this coincidence condition.

A comment is to the point here. Given three complex structures $I^a$ satisfying the Clifford algebra \p{Clifford},
 one can find, in its enveloping algebra, other triplets of complex structures forming the quaternion algebra \p{quaternion} \cite{GPS,Hull}.
For instance, the tensors
 \be
  \lb{J}
J^a=\frac12\epsilon^{abc}I^b I^c
 \ee
form the quaternion algebra \p{quaternion} and one might think that they define some HKT geometry.
This is not so, however, since there exists no connection with respect to which $J^a$ would be covariantly constant (since no such a connection
exists for $I^a$ in the CKT case).


\vspace{1mm}
It was shown in Ref. \cite{DI} that  the most general target HKT geometry  is reproduced in the framework
of the  superfield approach which is natural for physicists. In a generic case, one should use \cite{IL} the harmonic superspace approach \cite{GIOS},
but a wide class of the HKT metrics can be also obtained from an ordinary extended superspace  action
 describing the interactions of several {\it linear} root ({\bf 4, 4, 0})
multiplets of one sort \footnote{ The multiplet ({\bf 4, 4, 0}) is sometimes called {\it root} \cite{FG} because many other
${\cal N}=4, d=1$ multiplets can be obtained from it
 in the framework of the Hamiltonian reduction procedure \cite{BKMO} (or its Lagrangian version based on gauging some bosonic isometries \cite{DIgaug}).
For a detailed discussion, see recent \cite{taming}.
 Nonlinear ({\bf 4, 4, 0}) multiplets subject to some nonlinear constraints were defined and studied in \cite{DI}. Though they are necessary for constructing
 the general class of HKT ${\cal N}=4, d=1$ sigma models (including HK models as a subclass),  we will not come to grips
with this issue here.}.
 It was also demonstrated in \cite{DI} that ${\cal N}=4$ sigma models with more general  CKT target geometry
 arise,  when simultaneously including the so called ``mirror''  ({\bf 4, 4, 0}) multiplets, which
have slightly  different transformation laws under the ${\cal N}=4$ supersymmetry.
Thus, the minimal dimension of the bosonic target space of a CKT model is eight.
The corresponding general Lagrangian involving both types of the root multiplets  was constructed in \cite{DI} in terms of ${\cal N}=2$ superfields.
However, neither explicit component examples of the CKT systems, nor the relevant classical and quantum N\"other supercharges,  were given there.
On the other hand, the superfield and component Lagrangians for a particular
subclass of ${\cal N}=4, d=1$ sigma models with the set of two mutually mirror  linear ({\bf 4, 4, 0}) multiplets
were earlier constructed in \cite{ILS}. It was further shown there, that, under certain restrictions on the
superfield Lagrangian, this system possesses an extended ${\cal N}=8$ supersymmetry \cite{BIKL,ABC,Franc},
providing an example of the OKT  geometry. To better understand the interplay between HKT, CKT and OKT geometries,
one has to further elaborate  on these Lagrangians at the component level, to study the structure of the corresponding N\"other
supercharges  and to finally find out  how the differences between various target  geometries manifest themselves within this setting.

This is the basic purpose of the present paper.
To understand the difference between various geometries, we consider, in parallel with the model
based on two mutually mirror $({\bf 4, 4, 0})$  multiplets, also the system with two $({\bf 4, 4, 0})$ multiplets of the same kind.
We will explicitly show how the CKT structures arise in the first case, while the second system brings about the more familiar HKT geometry.

The pretty simple models with ${\cal N} = 4$ supersymmetry we are going to study here could also have a quite nice physical interpretation.
The bosonic target dimension of these ${\cal N} = 4$ models is eight and they can be constructed through switching on the appropriate mutual interaction
between two four-dimensional systems with HKT geometry. Since the latter models are known to describe the moduli space of five-dimensional black holes
with preservation of $1/4$ supersymmetry \cite{GPS}, the systems considered in the present paper can be interpreted as describing
moduli space of two interacting (in principle, non-identical) black holes, like as in ref. \cite{GutP2}.

\subsection{Supercharges in the $({\bf 4, 4, 0})$ SQM systems: a brief account}

Before giving the detailed exposition, here we explain the origin of different supersymmetric structures
 in general terms, without writing explicit formulas.

While considering these ({\bf 4, 4, 0}) multiplets separately, we have two sets of supersymmetry generators:
a set with the generators
$Q_{(1)}$ and $Q_{(1)}^a$ (where $a =1,2,3$) for one multiplet and
a set $Q_{(2)}$ and $Q_{(2)}^a$  for another multiplet
\footnote{For the reason to become clear later,
we keep manifest the diagonal $SU(2)$ in the product of two independent $SU(2)$ automorphism groups of ${\cal N}=4, d=1$
super Poincar\'e algebra. Respectively, the supercharges are divided into the singlet and triplet parts with respect to this explicit $SU(2)$.}.
Both sets $(Q_{(1)}$, $Q_{(1)}^a)$ and  $(Q_{(2)}$, $Q_{(2)}^a)$ form the ${\cal N}=4$ superalgebra on their own
and correspond to the HKT geometry.

Consider now a system   involving both multiplets.
At the first step, we may consider the free case.
Though it is almost trivial, the origin of different types of geometries in the considered
system can be understood already in this setting.
Indeed, the root multiplets do not involve auxiliary components. As a consequence, all supersymmetry transformations for the linear
multiplets are the same in both the free case and the interaction case.
We observe the following supersymmetries :
\begin{itemize}
\item Supersymmetry generated by
\begin{equation}
Q=Q_{(1)}+Q_{(2)}\,;
\end{equation}
\item Supersymmetry with the generators
\begin{equation}
\lb{Sa}
S^a=Q_{(1)}^a+Q_{(2)}^a\,;
\end{equation}
\item Supersymmetry generated by
\begin{equation}
\lb{Qa}
Q^a=Q_{(1)}^a-Q_{(2)}^a\,;
\end{equation}
\item A hidden ${\cal N}{=}\,4$ supersymmetry generated by certain supercharges
\begin{equation}
\lb{Qtilde}
\tilde Q\,,\qquad \tilde Q^a
\end{equation}
to be specified later.
The supersymmetry \p{Qtilde} mixes component fields from two different multiplets.
\end{itemize}

We may depict all  these supersymmetry algebras and their links to the different target geometries in the following way
\footnote{There is also the supersymmetry $S = Q_{(1)}- Q_{(2)}$ missing in the list above. But it is specific for the flat model whereas
the structures \p{sch} may survive, as we will see later, in interacting case.}
\begin{equation}\label{sch}
\lefteqn{\overbrace
{\phantom{\,\, S^{a}\quad\,\,\hat S\, Q \,\,}}^{{\cal N}=4\,({\rm HKT})}}
\,\, S^{a}\quad\quad \underbrace{\underbrace{\,\, Q \,\quad {\phantom{S_{A_A}}}\, Q^{a}\,}_{{\cal N}=4\,({\rm CKT})}
\quad\quad \tilde Q \quad\quad\,\, \tilde Q^{a}}_{{\cal N}=8\,({\rm OKT})}\,.
\end{equation}

Notice now that each supercharge in \p{Sa}, \p{Qa}, and \p{Qtilde} corresponds to some particular complex structure.
We will see that the structures corresponding to \p{Sa} form the quaternionic algebra \p{quaternion}, while the structures
corresponding to \p{Qa} do not satisfy \p{quaternion}, but only \p{Clifford}. Thus, the supercharges \p{Sa} are related to the HKT geometry while
the supercharges \p{Qa} --- to the CKT one.

For sure, in the non-interacting case, the metric is flat and the geometry is trivial. Still, the  distinction between
\p{Sa} and \p{Qa} that we observed helps
to understand what happens when the interaction between  two ${\cal N}{=}\,4$ multiplets
is switched on.
Actually, in an interacting system, only a part of the supersymmetries (\ref{sch}) survives. {\it Which} part --- it depends.
In the case of two  ${\cal N}{=}\,4$ multiplets of the same kind, there remains only ${\cal N}{=}\,4$ supersymmetry
with the generators  $Q$, $S^a$  revealing the target HKT geometry. In the case when one such ${\cal N}{=}\,4$ multiplet
interacts with the mirror  one,  we are left with the ${\cal N}{=}\,4$ supersymmetry
with the generators  $Q$, $Q^a$ and encounter the CKT geometry. In some cases, the supersymmetry involving both \p{Qa} and \p{Qtilde} persists,
and this is the OKT geometry.

\subsection{Plan}
The paper is organized as follows.

In Sect.\,2 we give the supersymmetric description of
${\cal N}{=}\,4$ $({\bf 4, 4, 0})$ multiplets, ordinary and mirror. We pass to the real component fields, which makes the picture more transparent.

In Sect.\,3,
we make precise the picture \p{sch} discussed above. To this end, we consider the systems with two non-interacting root multiplets,
of the same type and of different types. The component Lagrangians in these two cases are the same,
but the realizations of ${\cal N}=4$ supersymmetry are different.

In Sect.\,4, we consider nontrivial interacting models with two multiplets.
 We start with the  model involving two ordinary multiplets and describing an HKT sigma-model.
But our main goal is to describe the system with two mutually mirror multiplets
generating the CKT geometry. In particular, we rewrite this model in terms of ${\cal N}{=}\,2$ superfields
and find that the CKT  models provide a particular example of ${\cal N}{=}\,2$ superfield systems with external (anti)holomorphic
torsions which were considered in \cite{FIS}.
In contrast to HKT models, such models do not possess conserved fermion charges.
 We also show in which way, for certain special CKT models  with conformally flat 8-dimensional metric, the ${\cal N} = 4$
supersymmetry can be extended
to   ${\cal N} = 8\,$. This provides a nontrivial
 example of OKT geometry.

In Sect. 5, we construct the quantum supercharges for the models discussed in the preceding Sections and
demonstrate how this construction can be extended to a generic case of the CKT models.
This allows us to give a new geometric definition of the CKT (OKT) manifolds.

In Sect.\,6, some specific examples of 8-dimensional CKT and OKT manifolds are
discussed. Sect.\,7 is reserved for some concluding remarks.

\setcounter{equation}0
\section{(4, 4, 0) multiplets}
\subsection{Ordinary ({\bf 4, 4, 0}) multiplet}

We will use the ${\cal N}{=}\,4$ superspace with the coordinates $(t,\theta^{i{k}^\prime})$,
$(\overline{\theta^{i{k}^\prime}})= -\epsilon_{ij}\epsilon_{{k}^\prime{l}^\prime}\theta^{j{l}^\prime}   \equiv -\theta_{k'j} $.
The indices $i=1,2$ and $k^\prime=1,2$ are doublet indices of the ${\rm SU}_L(2)$ and ${\rm SU}_R(2)$ groups respectively, which
form the full automorphism group ${\rm SO}(4)={\rm SU}_L(2)\times{\rm SU}_R(2)$ of
the ${\cal N}{=}\,4$ superalgebra.
Covariant derivatives are
\begin{equation}\label{D}
D^{i{k}^\prime} = \frac{\partial}{\partial\theta_{{k}^\prime i}}+i\theta^{i{k}^\prime}\partial_t\,.
\end{equation}

The linear ({\bf 4, 4, 0}) multiplet is described by the (pseudo)real superfield $X^{i\alpha}(t,\theta)$,
$(\overline{X^{i\alpha}})=-\epsilon_{ij}\epsilon_{\alpha \beta}X^{j\beta} \equiv X_{\alpha i}$, subject to the constraints
\begin{equation}\label{const1}
D^{(i{j}^\prime}X^{k)\alpha} = 0\,
\end{equation}
(symmetrized over $i \leftrightarrow k$), where the index $\alpha=1,2$ is transformed by the additional Pauli-G\"{u}rsey ${\rm SU}(2)$ group
commuting with ${\cal N}{=}\,4$ supersymmetry transformations.

The solution of the off-shell constraints (\ref{const1}) reads
\begin{equation}\label{X}
X^{i\alpha} = x^{i\alpha}-\theta^{i{k}^\prime}\chi_{{k}^\prime}^{\alpha} +
i\theta^{i{k}^\prime}\theta_{{k}^\prime k }\dot x^{k\alpha} - {\textstyle\frac{i}{3}}\,\theta^{i{i}^\prime} \theta_{{i}^\prime k} \theta^{k{k}^\prime}
\dot\chi_{{k}^\prime}^{\alpha}
-{\textstyle\frac{1}{12}}\,\theta^{k{k}^\prime}  \theta_{{k}^\prime j}
\theta^{j{i}^\prime}  \theta_{{i}^\prime k} \, \ddot x^{\,i\alpha},
\end{equation}
and so it encompasses four real bosonic component  fields $(\overline{x^{i\alpha}})=-\epsilon_{ij}\epsilon_{\alpha \beta}x^{j\beta}$
and four real fermionic component fields
$(\overline{\chi^{i^\prime\alpha}})=-\epsilon_{i^\prime j^\prime}\epsilon_{\alpha \beta}\chi^{j^\prime\beta}$.

The superfield action
\begin{equation}\label{act1}
S_1 = \int dt d^4 \theta\, {\cal L}_1(X)\,,
\end{equation}
where $d^4 \theta \equiv - \frac{1}{24}\,D^{i{i}^\prime}D_{{i}^\prime k} D^{k{k}^\prime} D_{{k}^\prime i} $, yields the following component action
\footnote{It seems to have symmetry  $[SU(2)]^3$, but two of these $SU(2)$ symmetries (those realized on the indices
$i$ and $\alpha$) can be broken, since the factor $G_1$ is not required in general to respect any of them. One more $SU(2)$
(acting on the primed indices) affects only fermionic fields and is unbroken, even in the general actions of the (${\bf 4, 4, 0}$)
multiplets of the same sort \cite{DI}.}
\begin{equation}\label{act1-c}
S_1 = \int dt \,\Big[ -{\textstyle\frac{1}{2}}\,G_{1}\dot x^{i\alpha}\dot x_{\alpha i}+
{\textstyle\frac  {i}{4}}\,G_{1}  \dot\chi^{{i}^\prime\alpha} \chi_{\alpha i^\prime}+
i\,\dot x^i_\alpha R^{\alpha\beta} (\partial_{\beta i}G_{1}) +
{\textstyle\frac{1}{6}}\,(\triangle_x G_{1})\,R^{\alpha\beta}R_{\alpha \beta }\Big]\,.
\end{equation}
Here $G_{1}=\triangle_x{\cal L}_1(x)$ and
\begin{equation}\label{G1}
R^{\alpha\beta}  = R^{\beta \alpha}  =
 {\textstyle\frac{1}{4}}\,\chi_{{k}^\prime}^{\alpha}\chi^{{k}^\prime\beta}\, ,
\end{equation}
We use the following notations:
$\partial_{\alpha i}=\partial/\partial x^{i\alpha}$, $\triangle_x=\partial^2/\partial x^{i\alpha}\partial x_{\alpha i}$.

The N\"other charges of the supersymmetry transformations
\begin{equation}\label{susy-tr1}
\delta x^{i\alpha}= \varepsilon^{i{k}^\prime} \chi_{{k}^\prime}^{\alpha} \,,\qquad \delta\chi^{{i}^\prime\alpha}=-2i\varepsilon^{k{i}^\prime}
\dot x_k^{\alpha}
\end{equation}
produced by the standard realization of ${\cal N}=4$ supersymmetry in ${\cal N}=4$ superspace  are
\begin{equation}\label{susy-char1}
Q_i^{{j}^\prime}= \chi^{{j}^\prime\alpha } p_{\alpha i} + {\textstyle\frac{i}{12}}\,\,\
\chi^{{j}^\prime\beta } \chi_{\beta {k}^\prime} \chi^{{k}^\prime\alpha } (\partial_{\alpha i}G_{1}) \,,
\end{equation}
where $p_{\alpha i}=-G_{1}\dot x_{\alpha i}-i(\partial_{\beta i}G_{1})R_\alpha^\beta$ are canonical momenta of $x^{i\alpha}$.

Non-vanishing Poisson brackets are
\footnote{There are two ways to derive these expressions. First, one can introduce the tangent space
fermion variables and distinguish
carefully between $\psi^a$ (canonical coordinates) and $\bar\psi^a$ (canonical momenta)   such that
$\{\bar \psi^a,\, \psi^b\} = \delta^{ab}$ and all other brackets vanish. The variables $\chi^{i' \alpha}$ are expressed via
the flat fermion variables multiplied by the vielbeins, the derivatives of the latter bringing about nontrivial brackets
$\{\chi, p\}$.

 Alternatively, one can work with the original fermions
$\chi^{i' \alpha}$ carrying world indices. This
description involves a gauge-like redundancy in the  fermion sector: the Lagrangian involves only first derivatives
of $\chi^{i' \alpha}$ such that fermion canonical momenta are expressed via $\chi^{i' \alpha}$,
generating second class constraints. They can be resolved following the Dirac procedure and the expressions \p{DB-1} are derived
as  Dirac brackets.}
\begin{equation}\label{DB-1}
\{x^{i\alpha}, p_{\beta j}\}_{{}}=\delta^{i}_{j}\delta^{\alpha}_{\beta}\,,\qquad
\{\chi^{{i}^\prime\alpha }, \chi_{\beta {j}^\prime}\}_{{}}={\textstyle\frac{2i}{G_{1}}}\,\delta^{{i}^\prime}_{{j}^\prime}\delta^{\alpha}_{\beta}\,,\qquad
\{\chi^{{i}^\prime\alpha }, p_{\beta j}\}_{{}}=-{\textstyle\frac{1}{2G_{1}}}\,(\partial_{\beta j}G_{1})\,
\chi^{{i}^\prime \alpha} \,.
\end{equation}
Supercharges (\ref{susy-char1}) form the ${\cal N}=4$ supersymmetry algebra
\begin{equation}\label{susy-1}
\{Q_i^{{j}^\prime}, Q_k^{{l}^\prime}\}_{{}}=2i\,\epsilon^{{j}^\prime{l}^\prime}\epsilon_{ik}\,H_1 \,,
\end{equation}
where
\begin{equation}\label{H-1}
H_1=-{\textstyle\frac{1}{2G_{1}}}\,p^{i\alpha }p_{\alpha i}-
{\textstyle\frac{i}{G_{1}}}\,p^{k\alpha} R^{\beta}_{\alpha} (\partial_{\beta k}G_{1})
-{\textstyle\frac{1}{6}}\,\left[\Delta_xG_{1}- {\textstyle\frac{3}{2G_{1}}}\,
(\partial^{k\gamma}G_{1})(\partial_{\gamma k}G_{1})\right]R^{\alpha\beta}R_{\alpha \beta }
\end{equation}
is the canonical Hamiltonian for the system (\ref{act1-c}).

Let us now rewrite all expressions in terms of the real four-vector quantities
\begin{equation}\label{4vect-real1}
x^{A}=(\overline{x^{A}})\,,\qquad p_{A}=(\overline{p^{\phantom A}_{A}})\,,\qquad \chi^{A}=(\overline{\chi^{A}})\,,\qquad
A=1,2,3,4\,,
\end{equation}
defined by
\begin{equation}\label{4vect-1}
x_{\alpha i}={\textstyle\frac{1}{\sqrt{2}}}\,x^{A}(\sigma_A)_{\alpha i}\,,
\qquad x^{i\alpha}=-{\textstyle\frac{1}{\sqrt{2}}}\,x^{A}(\sigma^\dagger_A)^{i\alpha}\,,
\end{equation}
\begin{equation}\label{4vect-1-1}
p_{\alpha i}=-{\textstyle\frac{1}{\sqrt{2}}}\,p_{A}(\sigma^A)_{\alpha i}\,,\qquad
p^{i\alpha}={\textstyle\frac{1}{\sqrt{2}}}\,p_{A}(\sigma^{A\dagger})^{i\alpha}\,,
\end{equation}
\begin{equation}\label{4vect-1-1a}
\partial_{\alpha i}=-{\textstyle\frac{1}{\sqrt{2}}}\,\partial_{A}(\sigma^A)_{\alpha i}\,,\qquad
\partial^{i\alpha}={\textstyle\frac{1}{\sqrt{2}}}\,\partial_{A}(\sigma^{A\dagger})^{i\alpha}\,,
\end{equation}
\begin{equation}\label{4vect-1a}
\chi_{\alpha {i}^\prime}=\chi^{A}(\sigma_A)_{\alpha{i}^\prime}\,,\qquad
\chi^{{i}^\prime\alpha}=-\chi^{A}(\sigma^\dagger_A)^{{i}^\prime\alpha }\,.
\end{equation}
Here $\sigma^A=\sigma_A$ and $\sigma^{A\dagger}=\sigma^\dagger_A$ are defined as in \p{4sigma}.
We use the multiplier ${\textstyle\frac{1}{\sqrt{2}}}$ in (\ref{4vect-1})-(\ref{4vect-1-1a}) to ensure $x^{i\alpha}x_{\alpha i}=-x^{A}x^{A}$,
$\triangle_x=\partial^{i\alpha}\partial_{\alpha i}=-\partial_{A}\partial_{A}$, etc.
The opposite signs in (\ref{4vect-1}) and in (\ref{4vect-1-1}), (\ref{4vect-1-1a}) are chosen in order to preserve the equalities
$\{x^{A}, p_{B}\}=\delta^{A}_{B}$, $\partial_A x^B=\delta_{A}^{B}$.

In the new notation, the supercharges (\ref{susy-char1}) take the form
\begin{equation}\label{susy-char1-4}
Q_k^{{j}^\prime}= {\textstyle\frac{1}{\sqrt{2}}}\left[\chi^{A}p_{B}(\sigma^\dagger_{A}\sigma^{B})^{{j}^\prime}{}_k -
{\textstyle\frac{i}{12}}\,
\chi^{A}\chi^{B}\chi^{C}(\partial_{D}G_1)\,(\sigma^\dagger_{A}\sigma_{B}\sigma^\dagger_{C}\sigma^{D})^{{j}^\prime}{}_k \right].
\end{equation}
Splitting this expression into the singlet and the triplet parts,
\begin{equation}\label{susy-char1-s-tr}
Q={\textstyle\frac{1}{\sqrt{2}}}\,Q_k^{{j}^\prime} \delta_{{j}^\prime}^k\,,\qquad {Q}^a
={\textstyle\frac{i}{\sqrt{2}}}\,Q_k^{{j}^\prime}({\sigma}^a)^k{}_{{j}^\prime} \,,
\end{equation}
where $a =1, 2, 3$, we obtain
\begin{eqnarray}\label{susy-char1-s-1}
Q&=& \chi^{A}\,
\Big[p_{A} + {\textstyle\frac{i}{12}}\,\epsilon_{ABCD}(\partial_{D}G_1)\,\chi^{B}\chi^{C}\Big] \,,
\\
\label{susy-char1-tr-1}
Q^{{a}}&=& -\chi^E\,\eta^{{a}}_{EA}
\Big[p_{A} + {\textstyle\frac{i}{4}}\,\epsilon_{ABCD}(\partial_{D}G_1)\,\chi^{B}\chi^{C}\Big] \,,
\end{eqnarray}
where $\eta^{{a}}_{AB}$ are the 't Hooft symbols
(we remind their  definition and some properties in  Appendix~A).

The superfield $X^{i\alpha}$ can also be cast in the vector notations,
$X^{i\alpha} = -{\textstyle\frac{1}{\sqrt{2}}}\,X^{A}(\sigma^\dagger_A)^{i\alpha}$.
 Then
\be
\lb{XA}
X^A \ =\ x^A - \frac i2 \eta^a_{BC} \theta^B \theta^C \eta^a_{AD} \dot x^D +
\frac 1{24} \epsilon_{BCDE} \theta^B \theta^C \theta^D \theta^E \, \ddot x^A \ + \ {\rm fermion\ terms} \, ,
 \ee
where we used the definitions
 \be
\lb{thetavec}
 \theta^{i k^\prime}  \ =\ {\textstyle\frac{1}{\sqrt{2}}}\, (\sigma^\dagger_B)^{i k^\prime} \theta^B \, , \ \ \ \
 \theta_{ k^\prime i }  \ =\ - {\textstyle\frac{1}{\sqrt{2}}}\, (\sigma_B)_{k^\prime i} \theta^B\,.
 \ee
Note that the constraints \p{const1} imply in particular $\eta^a_{AB} D_A X^B = 0$ with
 \be
\lb{DA}
 D_A \ =\ \frac \partial {\partial \theta^A} - i \theta^A \frac {\partial} {\partial t} \, .
 \ee

In the four-vector notations, the action (\ref{act1-c}) takes the form
\begin{equation}\label{act1-c-g}
S_1 = \int dt \,\Big[{\textstyle\frac{1}{2}}\,g_{AB}\left(\dot x^{A}\dot x^{B}+i
\chi^{A}\hat\nabla\chi^{B}\right)-
{\textstyle\frac{1}{12}}\,\partial_A C_{BCD}\chi^{A}\chi^{B}\chi^{C}\chi^{D}\Big]\,,
\end{equation}
where
\begin{equation}\label{g1}
g_{AB}=G_1\delta_{AB}
\end{equation}
is the metric tensor, while
\begin{equation}\label{t1}
C_{ABC}=\epsilon_{ABCD}(\partial_DG_1)
\end{equation}
is the torsion. Covariant derivative of the fermionic fields
\begin{equation}\label{cd1}
\hat\nabla\chi^{A}=\dot\chi^{A}+\hat\Gamma^A_{BC}\dot x^B\chi^{C}
\end{equation}
involves the torsionful affine connection,
\begin{equation}\label{pc1}
\hat\Gamma_{A,BC}=g_{AD}\hat\Gamma^D_{BC}=\Gamma_{A,BC}+{\textstyle\frac{1}{2}}\,C_{ABC}\,,
\end{equation}
where
$\Gamma_{A,BC}={\textstyle\frac{1}{2}}\left(\partial_B g_{AC}+\partial_C g_{AB}-\partial_A g_{BC}\right)$
is the Levi-Civita connection.

The supercharges \p{susy-char1-s-1}, \p{susy-char1-tr-1} can be presented in a more geometric form, if introducing
the Lorentz spin connection
 \be
\lb{Omega}
\Omega_{A,BC}=e_B^{\underline{D}}\,e_C^{\underline{E}}\,\Omega_{A,\underline{D}\underline{E}} \,,
 \ \ \ \ \ \ \ \ \ \ \ \ \Omega_{A,\underline{B}\underline{C}}=e_{\underline{B}D}\left(\partial_A e_{\underline{C}}^{D} +
\Gamma^D_{AE}e_{\underline{C}}^{E} \right)
\ee
(tangent space indices are underlined).

In our case,
  the vierbein is $e^{\underline{B}}_{A}=\sqrt{G_1}\,\delta^{\underline{B}}_{A}$, and
\begin{equation}\label{Om-1}
\Omega_{A,BC}={\textstyle\frac12}\,\Big[\delta_{AB}(\partial_CG_1)-
\delta_{AC}(\partial_BG_1)\Big].
\end{equation}

We obtain
\begin{eqnarray}\label{susy-char1-s-2f}
Q&=& \chi^{A}
\Big(p_{A} -{\textstyle\frac{i}{2}}\,\Omega_{A,BC}\chi^{B}\chi^{C} +{\textstyle\frac{i}{12}}\,C_{ABC}\,\chi^{B}\chi^{C}\Big),
\\[6pt]
\label{susy-char1-tr-2f}
Q^{{a}}&=& \chi^D(I^{{a}})_{D}{}^{A}\Big(
p_{A} -{\textstyle\frac{i}{2}}\,\Omega_{A,BC}\chi^{B}\chi^{C} - {\textstyle\frac{i}{\,4\,}}\,C_{ABC}\,
\chi^{B}\chi^{C}\Big),
\end{eqnarray}
where the complex structure tensor was introduced,
\begin{equation}\label{compl-str-1}
(I^{{a}})_{A}{}^{B}=-\eta^{{a}}_{{A}B}\,.
\end{equation}

The representation  \p{susy-char1-s-2f}, \p{susy-char1-tr-2f} is valid for the supercharges in {\it any} HKT model \cite{QHKT}.
In the considered conformally flat 4-dimensional case, the second term  in (\ref{susy-char1-s-2f}) actually vanishes,
$-{\textstyle\frac{i}{2}}\,\Omega_{A,BC}\chi^{A}\chi^{B}\chi^{C}=0$.  Also note that, dividing the parameters
of supersymmetry transformations into the singlet and triplet parts,
\begin{equation}\label{eps-tr1a}
\varepsilon^{i}_{{k}^\prime}= -{\textstyle\frac{i}{\sqrt{2}}}\,\varepsilon\,\delta^{i}_{{k}^\prime} -
{\textstyle\frac{1}{\sqrt{2}}}\,\varepsilon_a(\sigma^a)^{i}{}_{{k}^\prime} \, ,
\end{equation}
we can represent the transformations (\ref{susy-tr1}) in the following form
\begin{equation}\label{susy-tr1a}
\delta x^{A}= i\,\varepsilon \chi^{A}+i \,\varepsilon_a(I^{{a}})^{A}{}_{B}\chi^{B}\,,\qquad
\delta \chi^{A}= -\varepsilon \dot x^{A}+ \varepsilon_a(I^{{a}})^{A}{}_{B} \, \dot x^{B}\,.
\end{equation}

The complex structure tensors (\ref{compl-str-1}) can be checked to form the quaternionic algebra \p{quaternion},
and they are covariantly constant with respect to the  connection \p{pc1}
\begin{equation}\label{const-str-1}
\hat\nabla_A(I^{{a}})_{B}{}^{C}=\partial_A(I^{{a}})_{B}{}^{C}-\hat\Gamma_{AB}^D(I^{{a}})_{D}{}^{C}+\hat\Gamma_{AD}^C(I^{{a}})_{B}{}^{D}=0\,.
\end{equation}
Thus, this model reproduces the HKT geometry as it should  and \p{pc1} is nothing but the Bismut connection with the torsion \p{Bismut},
which can also be expressed in this case as
\be
 C_{ABC} \ =\ \left( I_{AB} I_{CD} +  I_{AD} I_{BC} +
I_{AC} I_{DB}  \right)  \frac {\partial_D G_1} {G_1^2} \, .
 \ee

In the considered case, the supercharges can also be rewritten in a somewhat different form, using the
following relation,
\begin{equation}\label{equ-OC-4}
\chi^D(I^{{a}})_{D}{}^{A}\, \Omega_{A,BC}\chi^{B}\chi^{C}=
- \chi^D(I^{{a}})_{D}{}^{A}\,C_{ABC} \, \chi^{B}\chi^{C}\,.
\end{equation}
It allows us to rewrite  (\ref{susy-char1-tr-2f}) in various ways, in particular,  as
\begin{equation}\label{susy-char1a-tr-2f}
Q^{{a}}= \chi^D(I^{{a}})_{D}{}^{A}\Big(
p_{A} -{\textstyle\frac{i}{6}}\,\Omega_{A,BC}\chi^{B}\chi^{C} + {\textstyle\frac{i}{\,12\,}}\,C_{ABC}\,
\chi^{B}\chi^{C}\Big).
\end{equation}
Later, we will see that just the representation \p{susy-char1a-tr-2f}  allows a direct generalization to the CKT case. In the latter
case, there is no any analog of the relation \p{equ-OC-4} and it is impossible to cast the supercharges into the form \p{susy-char1-tr-2f}.

\subsection{Mirror ({\bf 4, 4, 0}) multiplet}

A mirror ({\bf 4, 4, 0}) multiplet is described by the superfield $Y^{ {i}^\prime \alpha^\prime}(t,\theta)$
subject to the constraints
\begin{equation}\label{const2}
D^{k({i}^\prime}Y^{{j}^\prime) \alpha^\prime } = 0\,,
\end{equation}
where the index $\alpha^\prime=1,2$ refers to an extra Pauli-G\"{u}rsey ${\rm SU}^\prime(2)$ group.
The constraints (\ref{const2}) are solved by
\begin{equation}
\label{Y}
Y^{{i}^\prime \alpha^\prime  } = y^{{i}^\prime \alpha^\prime }  -   \psi_{k}^{\alpha^\prime} \theta^{k{i}^\prime} -
i \dot y^{j^\prime \alpha^\prime  }
\theta_{{j}^\prime k } \theta^{k{i}^\prime} + {\textstyle\frac{i}{3}}\,\dot\psi_{j}^{\alpha^\prime} \theta^{j{k}^\prime}\theta_{{k}^\prime k}\theta^{k{i}^\prime}
+ {\textstyle\frac{1}{12}}\,\ddot y^{ {i}^\prime \alpha^\prime } \,
\theta^{k{k}^\prime} \theta_{{k}^\prime j} \theta^{j{j}^\prime} \theta_{{j}^\prime k}\,,
\end{equation}
i.e. they yield the same field contents, though with another assignment with respect to $SU(2)_L$ and $SU(2)_R$ automorphism groups.

The superfield action
\begin{equation}\label{act2}
S_2 = \int dt d^4 \theta\, {\cal L}_2(Y)
\end{equation}
results in the following component action
\begin{equation}\label{act2-c}
S_2 = \int dt \,\Big[ -{\textstyle\frac{1}{2}}\,G_2\dot y^{{i}^\prime \alpha^\prime }\dot y_{ \alpha^\prime {i}^\prime }+
{\textstyle\frac{i}{4}}\,G_2\dot\psi^{i \alpha^\prime  }\psi_{\alpha^\prime i} +
i\dot y^{{i}^\prime}_{\alpha^\prime}\,R^{\alpha^\prime\beta^\prime} (\partial_{\beta^\prime {i}^\prime}G_2)  +
{\textstyle\frac{1}{6}}\,(\triangle_yG_2)\,R^{\alpha^\prime\beta^\prime}R_{\alpha^\prime \beta^\prime }\Big].
\end{equation}
Here
\begin{equation}\label{G2}
G_2=-\triangle_y{\cal L}_2(y)\,, \quad R^{\alpha^\prime\beta^\prime}= {\textstyle\frac{1}{4}}\,\psi^{\alpha^\prime}_k
\psi^{k \beta^\prime}\,,
\end{equation}
and $\partial_{\alpha^\prime {i}^\prime}=\partial/\partial y^{{i}^\prime \alpha^\prime }\,$,
$\triangle_y=\partial^2/\partial y^{{i}^\prime \alpha^\prime }\partial y_{\alpha^\prime {i}^\prime}\,$.

The N\"other charges associated with the supersymmetry transformations
\begin{equation}\label{susy-tr2}
\delta y^{{i}^\prime \alpha^\prime }=  \psi_{k}^{\alpha^\prime}  \varepsilon^{k{i}^\prime} \,,\qquad
\delta\psi^{i \alpha^\prime }=2i\varepsilon^{i{k}^\prime} \dot y_{{k}^\prime}^{\alpha^\prime}
\end{equation}
read
\begin{equation}\label{susy-char2}
Q_i^{{j}^\prime}=
 p^{{j}^\prime \alpha^\prime }\psi_{ \alpha^\prime i}  -
{\textstyle\frac{i}{12}}\,\, (\partial^{ {j}^\prime \alpha^\prime}G_2) \,
\psi_{\alpha^\prime k} \psi^{k \beta^\prime } \psi_{\beta^\prime i } \, ,
\end{equation}
where $p_{\alpha^\prime {i}^\prime}=-G_2\dot y_{\alpha^\prime {i}^\prime}-
i\,R_{\alpha^\prime}^{\beta^\prime} (\partial_{\beta^\prime {i}^\prime}G_2)$ are the canonical momenta of $y^{ {i}^\prime \alpha^\prime}$.

Non-vanishing Poisson brackets are
\begin{equation}\label{DB-2}
\!\{y^{{i}^\prime \alpha^\prime }, p_{\beta^\prime {j}^\prime}\}_{{}}=\delta^{{i}^\prime}_{{j}^\prime}\delta^{\alpha^\prime}_{\beta^\prime} ,\quad
\{\psi^{i \alpha^\prime }, \psi_{ \beta^\prime j}\}_{{}}={\textstyle\frac{2i}{G_2}}\,\delta^{i}_{j}\delta^{\alpha^\prime}_{\beta^\prime} ,\quad
\{\psi^{i \alpha^\prime }, p_{\beta^\prime {j}^\prime }\}_{{}}=-{\textstyle\frac{1}{2G_2}}\,
(\partial_{\beta^\prime {j}^\prime}G_2)\,
\psi^{i \alpha^\prime }  .
\end{equation}
Supercharges (\ref{susy-char2}) form the ${\cal N}=4$ supersymmetry algebra
\begin{equation}\label{susy-2}
\{Q_i^{{j}^\prime}, Q_k^{{l}^\prime}\}_{{}}=2i\,\epsilon^{{j}^\prime{l}^\prime}\epsilon_{ik}\,H_2 \,,
\end{equation}
where
\begin{equation}\label{H-2}
H_2=-{\textstyle\frac{1}{2G_2}}\,p^{{i}^\prime \alpha^\prime }p_{\alpha^\prime {i}^\prime}-
{\textstyle\frac{i}{G_2}}\,  R^{\beta^\prime}_{\alpha^\prime} (\partial_{\beta^\prime {k}^\prime }{G_2})
p^{{k}^\prime \alpha^\prime }
-{\textstyle\frac{1}{6}}\,\left[\Delta_yG_2 - {\textstyle\frac{3}{2G_2}}\,
(\partial^{{k}^\prime \gamma^\prime }G_2)(\partial_{\gamma^\prime {k}^\prime}G_2)\right]R^{\alpha^\prime\beta^\prime}
R_{\alpha^\prime \beta^\prime }
\end{equation}
is the canonical Hamiltonian for the system (\ref{act2-c}).

Note that two $SU(2)$ symmetries (realized on the primed indices) are in general broken in the action \p{act2-c}, while one more $SU(2)$ realized
on the unprimed indices $i$ and acting only on fermions is unbroken. The interplay between these three $SU(2)$ symmetries is quite similar
to the one we had in the case of the action \p{act1-c}, up to the evident interchange between two automorphism $SU(2)$ symmetries of ${\cal N}=4$
superalgebra \p{susy-1}, \p{susy-2}.

Now we again pass to the real four-vector quantities
\begin{equation}\label{4vect-real2}
y^{M}=(\overline{y^{M}})\,,\qquad p_{M}=(\overline{p^{\phantom A}_{M}})\,,\qquad \psi^{M}=(\overline{\psi^{M}})\,,\qquad
M=1,2,3,4
\end{equation}
by the relations
\begin{equation}\label{4vect-2}
y_{\alpha^\prime {i}^\prime}={\textstyle\frac{1}{\sqrt{2}}}\,y^{M}(\sigma_M)_{\alpha^\prime{i}^\prime}\,,\quad
y^{{i}^\prime\alpha^\prime}=-{\textstyle\frac{1}{\sqrt{2}}}\,y^{M}(\sigma^\dagger_M)^{{i}^\prime\alpha^\prime }\,,
\end{equation}
\begin{equation}\label{4vect-2-1}
p_{\alpha^\prime {i}^\prime}=-{\textstyle\frac{1}{\sqrt{2}}}\,p_{M}(\sigma^M)_{\alpha^\prime{i}^\prime}\,,
\quad p^{{i}^\prime\alpha^\prime}={\textstyle\frac{1}{\sqrt{2}}}\,p_{M}(\sigma^{M\dagger})^{{i}^\prime\alpha^\prime }\,,
\end{equation}
\begin{equation}\label{4vect-2-1a}
\partial_{\alpha^\prime {i}^\prime}=-{\textstyle\frac{1}{\sqrt{2}}}\,\partial_{M}(\sigma^M)_{\alpha^\prime{i}^\prime}\,,
\quad \partial^{{i}^\prime\alpha^\prime}={\textstyle\frac{1}{\sqrt{2}}}\,\partial_{M}(\sigma^{M\dagger})^{{i}^\prime\alpha^\prime }\,,
\end{equation}
\begin{equation}\label{4vect-2a}
\psi_{\alpha^\prime i}=\psi^{M}(\sigma_{M})_{\alpha^\prime {i}}\,,\quad
\psi^{{i}\alpha^\prime}=-\psi^{{M}}(\sigma^\dagger_{M})^{{i}\alpha^\prime }\,.
\end{equation}

Here $\sigma^M \equiv \sigma_M$ and $\sigma^{M\dagger} \equiv \sigma^\dagger_M$.
One can also introduce the real  superfield $Y^M$, but, within the chosen conventions, the analog of Eq.\p{XA} does not look so
nice:
the flow of indices in \p{Y} is not so smooth.
One can choose other conventions, such that all the formulas for the mirror multiplet would be quite
 parallel to those for the ordinary one (see Appendix B). However,  these conventions are less convenient
for us because they  display the OKT structure (to be discussed in Sect. 3,4) in a more complicated manner.

The supercharges take the form
\begin{equation}\label{susy-char2-4}
Q_k^{{j}^\prime}= {\textstyle\frac{1}{\sqrt{2}}}\left[p_{M}\psi^{N}(\sigma^{M \dagger}\sigma_{N})^{{j}^\prime}{}_k +
{\textstyle\frac{i}{12}}\,(\partial_{M}G_2)\,
\psi^{N}\psi^{K}\psi^{L}(\sigma^{M\dagger }\sigma_{N}\sigma^\dagger_{K}\sigma_{L})^{{j}^\prime}{}_k \right].
\end{equation}
Dividing them into singlet and triplet parts
\begin{equation}\label{susy-char2-s-tr}
Q={\textstyle\frac{1}{\sqrt{2}}}\,Q_k^{{j}^\prime} \delta_{{j}^\prime}^k\,,\qquad
{Q}^a ={\textstyle\frac{i}{\sqrt{2}}}\,Q_k^{{j}^\prime}({\sigma}^a)^k{}_{{j}^\prime} \,,
\end{equation}
we obtain
\begin{eqnarray}\label{susy-char1-s-2}
Q&=& \psi^{M}\left(p_{M} +
{\textstyle\frac{i}{12}}\,\epsilon_{MNKL}(\partial_{L}{G_2})\,
\psi^{N}\psi^{K}\right),
\\[6pt]
\label{susy-char1-tr-2}
Q^{{a}}&=& \psi^{P}\eta^{{a}}_{PM}\left(p_{M} +
{\textstyle\frac{i}{4}}\,\epsilon_{MNKL}(\partial_{L}{G_2})\,
\psi^{N}\psi^{K}\right).
\end{eqnarray}
These expressions  have the same form as \p{susy-char1-s-tr}, up to the irrelevant sign of $Q^a$.

In the four-vector notations, the action (\ref{act2-c}) is rewritten as
\begin{equation}\label{act2-c-g}
S_1 = \int dt \,\Big[{\textstyle\frac{1}{2}}\,g_{MN}\left(\dot y^{M}\dot y^{N}+i
\psi^{M}\hat\nabla\psi^{N}\right)-
{\textstyle\frac{1}{12}}\,\partial_M C_{NKL}\psi^{M}\psi^{N}\psi^{K}\psi^{L}\Big],
\end{equation}
where the metric tensor and the torsion are defined as
\begin{equation}\label{g2-c}
g_{MN}={G_2}\,\delta_{MN}\,,
\qquad
C_{MNK}=\epsilon_{MNKL}(\partial_L{G_2})\,.
\end{equation}
The covariant derivative of the fermionic field is
\begin{equation}\label{G-c-g}
\hat\nabla\psi^{M}=\dot\psi^{M}+\hat\Gamma^M_{NK}\dot y^N\psi^{K}\,,\qquad
\hat\Gamma_{M,NK}=g_{ML}\hat\Gamma^L_{NK}=\Gamma_{M,NK}+{\textstyle\frac{1}{2}}\,C_{MNK}\,,
\end{equation}
where
$\Gamma_{M,NK}={\textstyle\frac{1}{2}}\left(\partial_N g_{MK}+\partial_K g_{MN}-\partial_M g_{NK}\right)$ are the standard Christoffel symbols.

The supercharges (\ref{susy-char1-s-2}), (\ref{susy-char1-tr-2})
can be rewritten in the canonical HKT form
\begin{eqnarray}\label{susy-char1-s-2a}
Q&=& \psi^{M}
\Big[\left(p_{M} -{\textstyle\frac{i}{2}}\,\Omega_{M,\underline{N}\underline{K}}\psi^{\underline{N}}\psi^{\underline{K}}\right)
+{\textstyle\frac{i}{12}}\,C_{MNK}\,\psi^{N}\psi^{K}\Big],
\\
\label{susy-char1-tr-2a}
Q^{{a}}&=& \psi^{L}(I^{{a}})_{L}{}^{M}\Big[
\left(p_{M} -{\textstyle\frac{i}{2}}\,\Omega_{M,\underline{N}\underline{K}}\psi^{\underline{N}}\psi^{\underline{K}}\right) -
{\textstyle\frac{i}{4}}\,C_{MNK}\,
\psi^{N}\psi^{K}\Big],
\end{eqnarray}
where $(I^{{a}})_{M}{}^{N}$ is the complex structure tensor,
\begin{equation}\label{compl-str-2}
(I^{{a}})_{M}{}^{N}=\eta^{{a}}_{MN}\,,
\end{equation}
$\Omega_{M,\underline{N}\underline{K}}=e_{\underline{N}L}\left(\partial_M e_{\underline{K}}^{L} +
\Gamma^D_{ML}e_{\underline{K}}^{L} \right)$ is the spin connection, and $\psi^{\underline{M}}=e^{\underline{M}}_{N}\psi^{N}\,$,
with $e^{\underline{M}}_{N}$ being the vierbeins. In our case $e^{\underline{M}}_{N}=\sqrt{G_2}\,\delta^{\underline{M}}_{N}$,
and the second terms  in (\ref{susy-char1-s-2a}) are identically zero, $\Omega_{M,NK}\psi^{M}\psi^{N}\psi^{K}=0\,$.

In other words, we have finally obtained {\it the same} model as for a single ordinary multiplet, up to an irrelevant sign in \p{compl-str-2}.
This sign, however, will play the important role in Sect. 4,  when treating the interaction between
the ordinary and mirror multiplets. That is the way how a nontrivial CKT geometry comes out.

\setcounter{equation}0
\section{Free system of two ({\bf 4, 4, 0}) multiplets}

Before studying nontrivial interacting systems,
we consider the free model living in 8-dimensional flat space to illustrate and make more precise
the picture \p{sch}  discussed in the Introduction. It will help us to understand what happens in the interacting case.

 The component action is
\begin{equation}\label{act-free-vec}
S_{free}= {\textstyle\frac{1}{2}}\,\int dt \,\Big[\dot x^A \dot x^A +
\dot y^M \dot y^M
+ i \chi^A \dot\chi^A +
i \psi^M \dot \psi^M \Big].
\end{equation}

It can be represented as a sum of two Lagrangians describing the flat ordinary and the flat mirror multiplets
\p{act1-c-g}, \p{act2-c-g} with $G_1 = G_2 = 1$.
In spinor notations,
\begin{equation}\label{act-free}
S_{free}= -{\textstyle\frac{1}{2}}\,\int dt \,\Big[\dot x^{i\alpha}\dot x_{\alpha i}+
\dot y^{{i}^\prime \alpha^\prime }\dot y_{\alpha^\prime {i}^\prime}
+{\textstyle\frac{i}{2}}\,\chi^{{i}^\prime\alpha}\dot\chi_{\alpha {i}^\prime}+
{\textstyle\frac{i}{2}}\,\psi^{i\alpha^\prime}\dot\psi_{\alpha^\prime i}\Big].
\end{equation}

The N\"other supercharges associated with the ${\cal N}=4$ supersymmetry transformations (\ref{susy-tr1}), (\ref{susy-tr2})
are expressed as
 \begin{equation}\label{susy-char-free}
Q_i^{{j}^\prime} = \chi^{j^\prime \beta }p^{(x)}_{\beta i} + p^{(y) {j}^\prime \beta^\prime }\psi_{\beta^\prime i} \,.
 \end{equation}
After passing to the four-vector notation according to (\ref{4vect-1})-(\ref{4vect-1a}), (\ref{4vect-2})-(\ref{4vect-2a}),
we obtain the equivalent set of supercharges
\begin{eqnarray}\label{susy-char-free-s}
Q&=& \chi^{A}p^{(x)}_{A} +\psi^{M} p^{(y)}_{M}  \,,
\\
\label{susy-char-free-tr}
Q^{{a}}&=& -\chi^{A}\eta^{{a}}_{{A}B} p^{(x)}_{B}  +\psi^{M}\eta^{{a}}_{{M}N} p^{(y)}_{N}  \,,
\end{eqnarray}
where  $Q = {\textstyle\frac{1}{\sqrt{2}}}Q_i^{{j}^\prime} \delta_{{j}^\prime}^i\,$ and
$Q^a :={\textstyle\frac{i}{\sqrt{2}}}\,Q_k^{{j}^\prime}({\sigma}^a)^k{}_{{j}^\prime}\,$.
The opposite signs in \p{susy-char-free-tr} imply the following block-diagonal form of the associated triplet of complex structures,
   \begin{equation}
\label{compl-str-free}
I^{{a}}=
\left(
\begin{array}{cc}
- \eta^{{a}}_{AB} & 0 \\
0 &  \eta^{{a}}_{MN} \\
\end{array}
\right).
\end{equation}

Now we observe that {\it the same} component action \p{act-free-vec} can  be obtained by considering
the system of two multiplets of the same sort, i.e. two free ordinary or two free mirror multiplets. The corresponding singlet N\"other supercharge
coincides with $Q$ defined in \p{susy-char-free-s}, while the triplet one is given by the expression
\begin{equation}\label{susy-char-free-tr-add}
S^{{a}}= -\chi^{A}\eta^{{a}}_{{A}B} p^{(x)}_{B}  -\psi^{M}\eta^{{a}}_{{M}N} p^{(y)}_{N}   \,,
\end{equation}
which differs from \p{susy-char-free-tr} by the relative sign of its two terms.

The supercharges  (\ref{susy-char-free-tr-add})
and (\ref{susy-char-free-s}) form the ${\cal N}=4$ superalgebra associated with the HKT complex structures
\begin{equation}\label{HKT-str-free}
J^{{a}}=
\left(
\begin{array}{cc}
- \eta^{{a}}_{{A}B} & 0 \\
0 & - \eta^{{a}}_{{M}N} \\
\end{array}
\right).
 \end{equation}
The complex structures \p{HKT-str-free} satisfy the quaternion algebra \p{quaternion}, as opposed to the structures \p{compl-str-free} which
constitute solely the Clifford algebra.

In the  considered free case, there is also a hidden ${\cal N}=4$ supersymmetry realized by the transformations
\begin{equation}\label{h-susy-tr1}
\delta x^{i\alpha}=   \eta^{\alpha {\beta}^\prime } \psi_{\beta^\prime}^{i}  \,,\qquad
\delta\chi^{{i}^\prime\alpha}=-2i  \eta^{\alpha {\beta}^\prime } \dot y_{{\beta}^\prime}^{{i}^\prime} \,,
\end{equation}
\begin{equation}\label{h-susy-tr2}
\delta y^{{i}^\prime \alpha^\prime} = -\chi_\beta^{i^\prime} \eta^{ \beta  {\alpha}^\prime}  \,,\qquad
\delta\psi^{i\alpha^\prime}=-2i \dot x_{{\beta}}^{i} \eta^{\beta {\alpha}^\prime }\,.
\end{equation}

The corresponding N\"other charges are
\begin{equation}\label{h-susy-char-free}
\tilde Q_\beta^{{\alpha}^\prime} = -\chi_{\beta k^\prime} p^{(y)}{}^{{k}^\prime \alpha^\prime }
+p^{(x)}_{  \beta k} \psi^{k{\alpha}^\prime} \,.
\nonumber
\end{equation}
In the four-vector notation this supercharge, as in the previous cases, can  equivalently be represented as the set of
singlet and triplet
supercharges,
\begin{equation}\label{h-susy-char-free-s-tr}
\tilde Q={\textstyle\frac{1}{\sqrt{2}}}\,Q_\beta^{{\alpha}^\prime} \delta^\beta_{{\alpha}^\prime}\,,\qquad \tilde
{Q}^p ={\textstyle\frac{i}{\sqrt{2}}}\,Q_\beta^{{\alpha}^\prime}({\sigma}^p)_{{\alpha}^\prime}{}^\beta \,, \; p=1, 2, 3\,,
\end{equation}
which have the following explicit form
\begin{eqnarray}\label{h-susy-char-free-s}
\tilde Q&=& -\chi^{A}\delta_{A}{}^{N}p^{(y)}_{N} +\psi^{M}\delta_{M}{}^{B} p^{(x)}_{B}    \,,
\\
\label{h-susy-char-free-tr}
\tilde Q^{{p}}&=& \chi^{A}\bar\eta^{{p}}_{{A}N} p^{(y)}_{N}  +\psi^{M}\bar\eta^{{p                               }}_{{M}B} p^{(x)}_{B}    \,,
\end{eqnarray}
Thus, we obtain four additional complex structure matrices
\begin{equation}\label{ad-compl-str-free}
\tilde I=
\left(
\begin{array}{cc}
0 & -\delta_{A}{}^{N} \\
\delta_{M}{}^{B} & 0 \\
\end{array}
\right), \qquad
\tilde I^{{p}}=
\left(
\begin{array}{cc}
0 & \bar\eta^{{p}}_{{A}N} \\
\bar\eta^{{p}}_{{M}B} & 0 \\
\end{array}
\right).
\end{equation}

The antisymmetric matrices $I^{{a}}$, $\tilde I$ and $\tilde I^{{p}}$ form  the
seven-dimensional Clifford algebra (when checking this, the identity \p{ide-imp} is handy).

\setcounter{equation}0
\section{Interaction of two ({\bf 4, 4, 0}) multiplets}

Consider now the system with the two mutually interacting root multiplets.

\subsection{Ordinary  multiplets}

When the system includes  several interacting
 multiplets ({\bf 4, 4, 0}) of the same type, its geometry  is HKT. Let us see how this comes about in the simplest nontrivial
case of two interacting ordinary multiplets.

Thus, we consider two pseudoreal superfields  $X^{i\alpha}(t,\theta), Z^{i\alpha}(t,\theta)$,
both of them
being subject to the constraints (\ref{const1}). The $\theta$ expansion of $X^{i\alpha}(t,\theta)$  was given in (\ref{X}),
for $Z^{i\alpha}(t,\theta)$ we have the expansion of the same type:
\begin{equation}\label{Z}
Z^{i\alpha} = z^{i\alpha}-\theta^{i{k}^\prime}\varphi_{{k}^\prime}^{\ \alpha} +
i\theta^{i{k}^\prime}\theta_{{k}^\prime k }\dot z^{k\alpha} -
{\textstyle\frac{i}{3}}\,\theta^{i{i}^\prime} \theta_{{i}^\prime k} \theta^{k{k}^\prime}
\dot\varphi_{{k}^\prime}^{\ \alpha}
-{\textstyle\frac{1}{12}}\,\theta^{k{k}^\prime}  \theta_{{k}^\prime j}
\theta^{j{i}^\prime}  \theta_{{i}^\prime k} \, \ddot z^{\,i\alpha},
\end{equation}
with four real bosonic component  fields in
$z^{i\alpha}=-\epsilon^{ij}\epsilon^{\alpha\beta}(\overline{z^{{\phantom{\delta}}}_{\beta j}})$
and four real fermionic component fields in
$\varphi^{i^\prime\alpha}=-\epsilon^{i^\prime j^\prime}\epsilon^{\alpha\beta}(\overline{\varphi^{{\phantom{\delta}}}_{j^\prime\beta}})\,$.

A general sigma-model superfield action,
\begin{equation}\label{act-2id}
S = \int dt d^4 \theta\, {\cal L}(X,Z)\,,
\end{equation}
amounts to the following component action \cite{IL}
\begin{eqnarray}\label{act-c-2id}
S&=& \int dt \,\Big( L_{b}+L_{2f}+L_{4f}\Big) ,\\[3pt]
L_{b}&=&
-{\textstyle\frac{1}{2}}\left(\triangle_x{\cal L}\right) \dot x^{i\alpha}\dot x_{\alpha i}
-{\textstyle\frac{1}{2}}\left(\triangle_z{\cal L}\right) \dot z^{i\alpha}\dot z_{\alpha i}
+2\epsilon^{ij}\left(\partial^{(x)}_{\alpha {i}}\partial^{(z)}_{\beta {j}}{\cal L}\right)
\dot x^{k\alpha}\dot z_{k}^{\beta} \,, \label{act-bose-2id}\\[6pt]
L_{2f}&=&   -{\textstyle\frac{i}{4}}\left(\triangle_x{\cal L}\right) \chi^{{i}^\prime\alpha}\dot\chi_{\alpha {i}^\prime}
-{\textstyle\frac{i}{4}}\left(\triangle_z{\cal L}\right) \varphi^{{i}^\prime\alpha}\dot\varphi_{\alpha {i}^\prime}
+ {\textstyle\frac{i}{2}}\,\epsilon^{ij}\left(\partial^{(x)}_{\alpha {i}}\partial^{(z)}_{\beta {j}}{\cal L}\right)
\left(\varphi_{{k}^\prime}^{\beta} \dot\chi^{{k}^\prime\alpha} - \dot\varphi_{{k}^\prime}^{\beta} \chi^{{k}^\prime\alpha}
\right)
\label{act-2fermi-2id}\\
&& +\,i \left(\partial^{(x)}_{\alpha k}{\triangle_x\cal L}\right)\dot x_{\beta}^k R_{(\chi)}^{\alpha\beta} +
2i\,\epsilon^{ij}\left(\partial^{(x)}_{ \alpha  i}\partial^{(z)}_{\beta j}\partial^{(z)}_{\gamma k}{\cal L}\right)
\dot x^{k\alpha}R_{(\varphi)}^{\beta\gamma} \nonumber \\
&& +\,{\textstyle\frac{i}{4}}\left(\partial^{(z)}_{\alpha k}{\triangle_x\cal L}\right)\dot x_{\beta}^k
\varphi_{j^\prime}^{\alpha}\chi^{j^\prime\beta}
- {\textstyle\frac{i}{2}}\,\epsilon^{ij}\left(\partial^{(z)}_{\alpha i}\partial^{(x)}_{\beta j}\partial^{(x)}_{\gamma k}{\cal L}\right)
\dot x^{k\beta} \varphi_{j^\prime}^{\alpha}\chi^{j^\prime\gamma} \nonumber \\
&& +\,i \left(\partial^{(z)}_{\alpha k}{\triangle_z\cal L}\right)\dot z_{\beta}^k R_{(\varphi)}^{\alpha\beta} +
2i\,\epsilon^{ij}\left(\partial^{(z)}_{\alpha i}\partial^{(x)}_{\beta j}\partial^{(x)}_{\gamma k}{\cal L}\right)
\dot z^{k\alpha}R_{(\chi)}^{\beta\gamma} \nonumber \\
&& +\,{\textstyle\frac{i}{4}}\left(\partial^{(x)}_{\alpha k}{\triangle_z\cal L}\right)\dot z_{\beta}^k \chi_{j^\prime}^{\alpha}\varphi^{j^\prime\beta}
- {\textstyle\frac{i}{2}}\,\epsilon^{ij}\left(\partial^{(x)}_{\alpha i}\partial^{(z)}_{\beta j}\partial^{(z)}_{\gamma k}{\cal L}\right)
\dot z^{k\beta} \chi_{j^\prime}^{\alpha}\varphi^{j^\prime\gamma}\,, \nonumber \\[6pt]
L_{4f} &=&
{\textstyle\frac{1}{6}}\,(\triangle^2_x{\cal L})\, R_{(\chi)}^{\alpha\beta}R_{(\chi)}{}_{\alpha \beta }+
{\textstyle\frac{1}{6}}\,(\triangle^2_z{\cal L})\, R_{(\varphi)}^{\alpha\beta}R_{(\varphi)}{}_{\alpha \beta }
\label{act-4fermi-2id}\\
&&
-\,{\textstyle\frac{1}{3}}\,
\epsilon^{ij}\left(\partial^{(x)}_{\alpha i}\partial^{(z)}_{\beta j}\triangle_x{\cal L}\right) R_{(\chi)}^{\alpha\gamma} \chi_{{k}^\prime\gamma} \varphi^{\beta k^\prime }
+ {\textstyle\frac{1}{3}}\,\epsilon^{ij}\left(\partial^{(x)}_{\alpha i}\partial^{(z)}_{\beta j}\triangle_z{\cal L}\right)
R_{(\varphi)}^{\beta\gamma} \varphi_{\gamma {k}^\prime} \chi^{\,k^\prime\alpha} \nonumber\\
&& -\,{\textstyle\frac{1}{2}}\,(\triangle_x\triangle_z{\cal L}) R_{(\chi)}^{i^\prime j^\prime} R_{(\varphi)}{}_{ i^\prime j^\prime}
+2\epsilon^{ik}\epsilon^{jl} \left(\partial^{(x)}_{\alpha i}\partial^{(x)}_{\beta j}\partial^{(z)}_{\gamma k }\partial^{(z)}_{\delta l}{\cal L}\right)
R_{(\chi)}^{\alpha\beta}R_{(\varphi)}^{\gamma\delta}\,, \nonumber
\end{eqnarray}
where
\be
\lb{Riprime}
R_{(\chi)}^{i^\prime j^\prime} = {\textstyle\frac{1}{4}}\, \chi^{ i' \gamma} \chi^{j'}_\gamma, \ \ \ \ \ \ \
R_{(\varphi) \, i^\prime j^\prime} = {\textstyle\frac{1}{4}}\, \varphi_{i'}^\gamma \varphi_{\gamma j'} \, .
  \ee
The action (\ref{act-c-2id}) is invariant with respect to ${\cal N}=\,4$ supersymmetry transformations (\ref{susy-tr1}) of $x^{i\alpha}$, $\chi^{{i}^\prime\alpha}$
and similar transformations for the component fields of the superfield $Z^{i\alpha}$:
\begin{equation}\label{susy-tr-add-a}
\delta x^{i\alpha}= \xi^{i{k}^\prime} \chi_{{k}^\prime}^{\alpha} \,,\quad
\delta\chi^{{i}^\prime\alpha} = -2i \dot x_k^{\alpha}  \xi^{k{i}^\prime} \, ;\qquad
\delta z^{i\alpha}= \xi^{i{k}^\prime} \varphi_{{k}^\prime}^{\alpha} \,,\quad
\delta\varphi^{{i}^\prime\alpha}=-2i \dot z_k^{\alpha} \xi^{k{i}^\prime} \,\,.
\end{equation}
It is ${\cal N}=\,4$ HKT supersymmetry with complex structures $J^a$ defined in (\ref{HKT-str-free}).

Note that the bosonic part of the action generically includes  the mixed kinetic term $\propto \dot{x} \dot{z}$.
It vanishes only under
the condition
\begin{equation}\label{cond-2id}
\epsilon^{ij}\left(\partial^{(x)}_{\alpha i}\partial^{(z)}_{\beta j}{\cal L}\right)=0\,,
\end{equation}
which holds  if and only if the Lagrangian is a sum of two
terms depending, respectively, only on $x, \chi$ and only on $z,\varphi$.
In other words, two multiplets do not interact in this case\footnote{Clearly, \p{cond-2id} is necessary for that.
Let us prove that it is also sufficient. Note first that all other mixed terms in \p{act-2fermi-2id} and \p{act-4fermi-2id}
are reduced to \p{cond-2id} and its derivatives. There are also the terms involving $\triangle_x {\cal L} \equiv G_1$, \
$\triangle_z {\cal L} \equiv G_2$,
and their
derivatives. To understand their structure, let us act on  \p{cond-2id} with $\partial^{\alpha (x)}_k
= \epsilon^{\alpha \gamma} \partial^{ (x)}_{\gamma k}$. The operator $ \epsilon^{\alpha \gamma} \partial^{ (x)}_{\gamma k}
\partial^{ (x)}_{\alpha  i}$ is antisymmetric with respect to $k \leftrightarrow i$ and so is reduced to $(1/2) \epsilon_{ki} \triangle_x$.
We see that $\partial^{(z)}_{\beta j} G_1 = 0$, i.e. $G_1$ does not depend on $z$.
By the same token, $G_2$ does not depend on $x$. Bearing in mind the remarks above,
this proves our assertion.}.

It is also worth noting that the action \p{act-c-2id} respects the invariance under one of the automorphism $SU(2)$ symmetries
(the one acting on the primed indices $i'$, i.e. realized only on the fermionic filds), like the actions of single (${\bf 4, 4, 0}$)
multiplets.

\subsection{Two mutually mirror (4, 4, 0) multiplets}

We now consider the  superfield action
\begin{equation}\label{act}
S = \int dt d^4 \theta\, {\cal L}(X,Y)\,,
\end{equation}
where $X$ is an ordinary multiplet, while $Y$ is a mirror one.
This action gives rise to the following component form \cite{ILS}
\begin{eqnarray}\label{act-c-eq}
S&=& \int dt \, L=\int dt \,\Big( L_{b}+L_{2f}+L_{4f}\Big) ,\\[3pt]
L_{b} &=&
-{\textstyle\frac{1}{2}}\,G_1\, \dot x^{i\alpha}\dot x_{\alpha i}
-{\textstyle\frac{1}{2}}\,G_2\, \dot y^{ i^\prime \alpha^\prime} \dot y_{\alpha^\prime i^\prime}
\,, \label{act-bose-eq}\\[6pt]
L_{2f}&=&   -{\textstyle\frac{i}{4}}\,G_1\, \chi^{{i}^\prime\alpha}\dot\chi_{\alpha i^\prime}
-{\textstyle\frac{i}{4}}\,G_2\, \psi^{{i}\alpha^\prime} \dot\psi_{\alpha^\prime i}
\label{act-2fermi-eq}\\
&& +\,i \left(\partial^{(x)}_{\alpha k}\,G_1\right)\dot x_{\beta}^k R^{\alpha\beta} -
i \left(\partial^{(x)}_{\alpha i}\,G_2\right) R^{ik} \dot x^{\alpha}_k
- {\textstyle\frac{i}{2}}\left(\partial^{(y)}_{ \alpha^\prime i^\prime}\,G_1\right) \chi^{i^\prime\alpha}
\dot x_{\alpha i} \psi^{i\alpha^\prime} \nonumber \\
&& +\, i \left(\partial^{(y)}_{\alpha^\prime k^\prime }\,G_2\right)\dot y_{\beta^\prime}^{k^\prime}
 R^{\alpha^\prime\beta^\prime} -
i  \left(\partial^{(y)}_{ \alpha^\prime i^\prime }\,G_1\right)  R^{i^\prime k^\prime} \dot y^{\alpha^\prime}_{k^\prime}
- {\textstyle\frac{i}{2}}\left(\partial^{(x)}_{\alpha i}\,G_2\right)  \psi^{i\alpha^\prime}
\dot y_{\alpha^\prime i^\prime }\chi^{i^\prime\alpha}\,, \nonumber \\[6pt]
L_{4f} &=&
{\textstyle\frac{1}{6}}\left(\triangle_x\,G_1\right) R^{\alpha\beta}R_{\beta \alpha}+
{\textstyle\frac{1}{6}}\left(\triangle_y\,G_2\right) R^{\alpha^\prime\beta^\prime}R_{\beta^\prime \alpha^\prime}\label{act-4fermi-eq}\\
&& -\,{\textstyle\frac{1}{3}}\left(\partial^{(x)}_{\alpha i}\partial^{(y)}_{ \alpha^\prime i^\prime}\,G_1\right) R^{\alpha\beta}\chi^{{i}^\prime}_\beta\psi^{i\alpha^\prime}
- {\textstyle\frac{1}{3}}\left(\partial^{(x)}_{\alpha i}\partial^{(y)}_{\alpha^\prime {i}^\prime}\,G_2\right) R^{\alpha^\prime\beta^\prime}\psi^i_{\beta^\prime}\chi^{{i}^\prime\alpha} \nonumber\\
&&+\, \left(\partial^{(x)}_{\alpha i}\partial^{(x)}_{\beta j}\,G_2\right) R^{\alpha\beta}R^{ij}
+ \left(\partial^{(y)}_{\alpha^\prime i^\prime }\partial^{(y)}_{ \beta^\prime j^\prime}\,G_1\right)
R^{\alpha^\prime\beta^\prime}R^{{i}^\prime{j}^\prime} \,. \nonumber
\end{eqnarray}
Here
\begin{equation}\label{MN8}
G_1(x,y)=\triangle_x{\cal L}(x,y)\,, \quad G_2(x,y)=-\triangle_y{\cal L}(x,y)\,,
\end{equation}
$R^{\alpha \beta}$, $R^{\alpha' \beta'}$, $R^{i' j'}$ were defined above, and
\begin{equation}\label{Om34}
R^{ij}=  R^{ji} = {\textstyle\frac{1}{4}}\,\psi^{i \gamma^\prime } \psi^j_{\gamma^\prime}\,.
\end{equation}

Note that, in contrast to the Lagrangian \p{act-bose-2id} describing the interaction of two ordinary $({\bf 4,4,0})$  multiplets,
the mixed kinetic terms
$\propto \dot x \dot y$ are absent here. This can be best understood when using the conventions of Appendix B (see \p{nomixed1}, \p{nomixed2}).

The N\"other supercharges corresponding to the ${\cal N}{=}\,4$ supersymmetry transformations (\ref{susy-tr1}), (\ref{susy-tr2})
are
\begin{eqnarray}\label{susy-char}
Q_i^{{j}^\prime} &=& \chi^{{j}^\prime\alpha } p^{(x)}_{ \alpha i} + p^{(y)}{}^{ j^\prime \alpha^\prime}\psi_{\alpha^\prime i}\\[6pt]
&&
+\,{\textstyle\frac{i}{12}}  \chi^{{j}^\prime\beta } \chi_{\beta {k}^\prime }  \chi^{{k}^\prime\alpha }
\left(\partial^{(x)}_{\alpha i}\,G_1\right)
- {\textstyle\frac{i}{12}}\left(\partial^{(y) j^\prime \alpha^\prime }\,G_2\right)  \psi_{\alpha^\prime k} \psi^{k\beta^\prime}
\psi_{\beta^\prime i}
\nonumber\\[6pt]
&& + \,{\textstyle\frac{i}{4}}  \chi^{{j}^\prime\beta } \chi_{\beta {k}^\prime}
\left(\partial^{(y)}{}^{{k}^\prime \alpha^\prime }\,G_1\right) \psi_{\alpha^\prime i}
- {\textstyle\frac{i}{4}} \chi^{{j}^\prime\beta } \left(\partial^{(x)}_{{\beta} k}\,G_2\right)
\psi^{k\alpha^\prime} \psi_{\alpha^\prime i}
 \,.
\nonumber
\end{eqnarray}

In the case without the $X,Y$ interaction, i.e. for ${\cal L}(X,Y) = A(X) + B(Y)$, the action (\ref{act-c-eq}) coincides with the action
\p{act-bose-2id} in which the multiplets $X$ and $Y$ (substituted for $Z$) do not interact.  Such  action, similar to the free action \p{act-free},
is invariant with respect to two different
sets of supersymmetry transformations: the CKT transformations  (\ref{susy-tr1}), (\ref{susy-tr2})
 and the HKT transformations \p{susy-tr-add-a} (in which one should make replacements $z \rightarrow y,\, \varphi \rightarrow \chi$).
 But, generically, when the mutual interactions are switched on, the actions (\ref{act-c-eq}) and \p{act-bose-2id}
are different. The action (\ref{act-c-eq}) is invariant with respect to the CKT transformations, but not with respect to the HKT ones.
For the action \p{act-bose-2id}, the inverse is true.

An important new feature of the mixed action \p{act-c-eq} is that in general it breaks all four $SU(2)$ symmetries one can realize on the
component fields. The reason is that the metric functions $G_1$ and $G_2$ depend now on the bosonic fields carrying the doublet indices of
all these $SU(2)$ symmetries which so can be totally broken.

We have seen above that the free action \p{act-free} enjoys an additional invariance under the transformations
\p{h-susy-tr1}, (\ref{h-susy-tr2}) which  mix components from different multiplets.
The action \p{act-c-eq} also has this property under the condition
\begin{equation}\label{GG}
G_1(x,y)=G_2(x,y)\equiv G(x,y)\,,
\end{equation}
when the metric becomes conformally flat.
In this case, the ${\cal N} = 4$ CKT supersymmetry extends to the ${\cal N} = 8$ OKT supersymmetry.
The N\"other supercharges corresponding to the transformations \p{h-susy-tr1}, (\ref{h-susy-tr2}) are
\begin{eqnarray}
\label{h-susy-char}
\tilde Q_\beta^{{\alpha}^\prime} &=& -\chi_{\beta k^\prime} p^{(y)}{}^{{k}^\prime \alpha^\prime }
+p^{(x)}_{\beta k} \psi^{k \alpha^\prime}\\[6pt]
&&
- \,{\textstyle\frac{i}{12}} \left(\partial^{(x)}_{\beta i}\,G\right)
\psi^{i\gamma^\prime}\psi_{\gamma^\prime k}  \psi^{k\alpha^\prime}
- {\textstyle\frac{i}{12}}  \chi_{\beta i^\prime} \chi^{i^\prime\gamma} \chi_{\gamma k^\prime}
\left(\partial^{(y)}{}^{ k^\prime \alpha^\prime}\,G\right)
\nonumber\\[6pt]
&& +\,{\textstyle\frac{i}{4}} \chi_{\beta k^\prime} \left(\partial^{(y)}{}^{{k}^\prime \gamma^\prime }\,G\right)
\psi_{\gamma^\prime i} \psi^{i\alpha^\prime}
+{\textstyle\frac{i}{4}} \chi_{\beta i^\prime} \chi^{i^\prime\gamma} \left(\partial^{(x)}_{\gamma k}\,G\right)
\psi^{k\alpha^\prime} \,.
\nonumber
\end{eqnarray}

It is important to mention that the condition (\ref{GG})
together with the relations (\ref{MN8}) lead to the $D=8$ harmonicity of the Lagrangian,
\begin{equation}\label{har-L}
(\triangle_x+\triangle_y)\,{\cal L}(x,y)=0\,,
\end{equation}
as well as to the $D= 8$ harmonicity of the conformal factor
\begin{equation}\label{har-G}
(\triangle_x+\triangle_y)\,G(x,y)=0\,.
\end{equation}

\subsection{Eight-dimensional formulation}

It is instructive to rewrite the relations of the previous subsection  in the eight-dimensional vector notations.

Let us introduce 8-component quantities composed out of the 4-component ones:
\begin{equation}\label{8-comp-b}
x^{\mu}=\left(x^{A},y^{M} \right),\qquad
\psi^{\mu}=\left(\chi^{A},\psi^{M} \right),\qquad \mu=1,\ldots,8\,.
\end{equation}
Then, the 8-dimensional metric is
\begin{equation}
\label{g8}
g_{\mu\nu}=
\left(
\begin{array}{cc}
G_1\,\delta_{AB} & 0 \\
0 & G_2\,\delta_{MN} \\
\end{array}
\right),
\end{equation}
and the Lagrangian (\ref{act-c-eq}) takes the concise form
\begin{equation}
\label{act-c-g}
L={\textstyle\frac{1}{2}}\,g_{\mu\nu}\left(\dot x^{\mu}\dot x^{\nu}+i
\psi^{\mu}\hat\nabla\psi^{\nu}\right)-
{\textstyle\frac{1}{12}}\,\partial_\mu C_{\nu\lambda\rho}\psi^{\mu}\psi^{\nu}\psi^{\lambda}\psi^{\rho}\,,
\end{equation}
where
\begin{equation}\label{nabla8}
\hat\nabla\psi^{\mu}=\dot\psi^{\mu}+\hat\Gamma^\mu_{\nu\lambda}\dot x^\nu\psi^{\lambda}
\,,\qquad
\hat\Gamma_{\mu,\nu\lambda}=g_{\mu\rho}\hat\Gamma^\rho_{\nu\lambda}=\Gamma_{\mu,\nu\lambda}+{\textstyle\frac{1}{2}}\,C_{\mu\nu\lambda}\,.
\end{equation}
It implies the following
non-vanishing Poisson brackets for the (curved) phase space variables
\begin{equation}\label{DB-8}
\{x^{\mu}, p_{\nu}\}_{{}}=\delta^{\mu}_{\nu}\,,\qquad
\{\psi^{\mu}, \psi^{\nu}\}_{{}}=-ig^{\mu\nu}\,,\qquad
\{\psi^{\mu}, p_{\nu}\}_{{}}=-{\textstyle\frac{1}{2}}\,g^{\mu\lambda}(\partial_{\nu}g_{\lambda\rho})\,\psi^{\rho} \,.
\end{equation}

The Lagrangian \p{act-c-g} looks very similar to the HKT Lagrangian \p{act2-c-g}. However, there is an important difference between both.
While in the HKT case the connection \p{pc1} had a transparent geometric interpretation as the Bismut connection for the HKT complex structures,
the connection \p{nabla8} cannot be interpreted as a Bismut connection. It is something else. We will come back to this point later.

The Levi-Civita connection
$\Gamma_{\mu,\nu\lambda}={\textstyle\frac{1}{2}}\left(\partial_\nu g_{\mu\lambda}+\partial_\lambda
g_{\mu\nu}-\partial_\mu g_{\nu\lambda}\right)$
has the following non-vanishing components
\begin{equation}\label{Gamma1}
\Gamma_{A,BC}={\textstyle\frac12}\left[\delta_{AB}(\partial_C^{(x)}G_1)+
\delta_{AC}(\partial_B^{(x)}G_1)-\delta_{BC}(\partial_A^{(x)}G_1) \right],
\end{equation}
\begin{equation}\label{Gamma2}
\Gamma_{M,NK}={\textstyle\frac12}\left[\delta_{MN}(\partial_K^{(y)}G_2)+
\delta_{MK}(\partial_N^{(y)}G_2)-\delta_{NK}(\partial_M^{(y)}G_2) \right],
\end{equation}
\begin{equation}\label{Gamma3}
\Gamma_{M,NA}=-\Gamma_{A,MN}={\textstyle\frac12}\,\delta_{MN}(\partial_A^{(x)}G_2)\,,\qquad
\Gamma_{A,BM}=-\Gamma_{M,AB}={\textstyle\frac12}\,\delta_{AB}(\partial_M^{(y)}G_1)\,.
\end{equation}
The non-vanishing components of the torsion are
\begin{equation}\label{tors1}
C_{ABC}=\epsilon_{ABCD}(\partial_D^{(x)}G_1)\,,\qquad C_{MNK}=\epsilon_{MNKL}(\partial_L^{(y)}G_2)\,,
\end{equation}
\begin{equation}\label{tors2}
C_{MAB}=-C_{AMB}=C_{ABM}= - \eta^a_{AB} \eta^a_{MK} (\partial_K^{(y)}G_1)\,,
\end{equation}
\begin{equation}\label{tors3}
C_{AMN}=-C_{MAN}=C_{MNA}= - \eta^a_{AC} \eta^a_{MN} (\partial_C^{(x)}G_2)\,.
\end{equation}
Note that the torsion components (\ref{tors2}), (\ref{tors3}) obey the
$4D$ self-duality conditions
\begin{equation}\label{sd-tors}
\epsilon_{ABCD}C_{MCD}=2\,C_{MAB}\,,\qquad \epsilon_{MNKL}C_{AKL}=2\,C_{AMN}\,.
\end{equation}

The standard spin connection, $\Omega_{\mu,{\nu}{\lambda}}=e_\nu^{\underline{\rho}}\,e_\lambda^{\underline{\sigma}}\,
\Omega_{\mu,\underline{\rho}\underline{\sigma}}$,
$\Omega_{\mu,\underline{\nu}\underline{\lambda}}=e_{\underline{\nu}\rho}\left(\partial_\mu e_{\underline{\lambda}}^{\rho} +
\Gamma^\rho_{\mu\sigma}e_{\underline{\lambda}}^{\sigma} \right),$ for the metric \p{g8} is reduced  to
\begin{equation}\label{Om8}
\Omega_{\mu,{\nu}{\lambda}}=\Gamma_{\nu,{\mu}{\lambda}}-{\textstyle\frac12}\,\partial_\mu g_{\nu\lambda}= \partial_{[\lambda} g_{\nu]\mu}\,,
\end{equation}
that is a collection of the following components
\begin{equation}\label{Om1}
\Omega_{A,BC}={\textstyle\frac12}\left[\delta_{AB}(\partial_C^{(x)}G_1)-
\delta_{AC}(\partial_B^{(x)}G_1)\right],\quad
\Omega_{M,NK}={\textstyle\frac12}\left[\delta_{MN}(\partial_K^{(y)}G_2)-
\delta_{MK}(\partial_N^{(y)}G_2)\right],
\end{equation}
\begin{equation}\label{Om2}
\Omega_{M,NA}=-\Omega_{M,AN}={\textstyle\frac12}\,\delta_{MN}(\partial_A^{(x)}G_2)\,,\qquad
\Omega_{A,BM}=-\Omega_{A,MB}={\textstyle\frac12}\,\delta_{AB}(\partial_M^{(y)}G_1)\,,
\end{equation}
\begin{equation}\label{Om3}
\Omega_{M,AB}=0\,,\qquad
\Omega_{A,MN}=0\,.
\end{equation}
These expressions imply that
\be
\lb{Ompsi30}
\Omega_{\mu,\nu\lambda}\psi^{\mu}\psi^{\nu}\psi^{\lambda}=0\, .
 \ee

The canonical classical Hamiltonian of the system (\ref{act-c-g}) has the following form,
\begin{equation}\label{ham-c-g}
H={\textstyle\frac{1}{2}}\,g^{\mu\nu}{\cal P}_\mu{\cal P}_\nu+
{\textstyle\frac{1}{12}}\,\partial_\mu C_{\nu\lambda\rho}\psi^{\mu}\psi^{\nu}\psi^{\lambda}\psi^{\rho}\,,
\end{equation}
where
\begin{equation}\label{mom-c-g}
{\cal P}_\mu=p_\mu- {\textstyle\frac{i}{2}}\,\hat\Omega_{\mu,\nu\lambda}\psi^{\nu}\psi^{\lambda}\,,
\end{equation}
and $\hat\Omega_{\mu,\nu\lambda} $ is the torsionful spin connection,
\begin{equation}
\label{hat-om}
\hat\Omega_{\mu,\nu\lambda}=\Omega_{\mu,\nu\lambda}- {\textstyle\frac{1}{2}}\,C_{\mu\nu\lambda}\,.
\end{equation}
It corresponds (notwithstanding an ostensibly  opposite sign of the second term!) to the torsionful affine
 connection (\ref{nabla8}),
$$\hat\Omega_{\mu,{\nu}{\lambda}}=e_\nu^{\underline{\rho}}e_\lambda^{\underline{\sigma}}\,
\hat\Omega_{\mu,\underline{\rho}\underline{\sigma}}\,, \ \ \ \ \ \ \ \ \ \ \ \
\hat\Omega_{\mu,\underline{\nu}\underline{\lambda}}=e_{\underline{\nu}\rho}\left(\partial_\mu e_{\underline{\lambda}}^{\rho} +
\hat\Gamma^\rho_{\mu\sigma}e_{\underline{\lambda}}^{\sigma} \right).$$
Note that, in the considered case, the equality
$
\Omega_{\mu,\nu\lambda}\psi^{\nu}\psi^{\lambda}=\Gamma_{\nu,\mu\lambda}\psi^{\nu}\psi^{\lambda}
$
is satisfied.

The singlet and triplet parts of the supercharges (\ref{susy-char}) take the form
\begin{eqnarray}\label{susy-char1-s-2g}
Q&=& \psi^{\mu}
\Big(p_{\mu} -{\textstyle\frac{i}{2}}\,\Omega_{\mu,\nu\lambda}\psi^{\nu}\psi^{\lambda} +{\textstyle\frac{i}{12}}\,C_{\mu\nu\lambda}\,\psi^{\nu}\psi^{\lambda}\Big) ,
\\[6pt]
\label{susy-char1-tr-2g}
Q^{{a}}&=& \psi^{\sigma}(I^{{a}})_{\sigma}{}^{\mu}\Big(
p_{\mu} -{\textstyle\frac{i}{6}}\,\Omega_{\mu,\nu\lambda}\psi^{\nu}\psi^{\lambda}+{\textstyle\frac{i}{12}}\,C_{\mu\nu\lambda}\,\psi^{\nu}\psi^{\lambda}\Big),
\end{eqnarray}
with the  block--diagonal complex structures $(I^a)_\sigma^{\ \mu}$ given by the expressions \p{compl-str-free}.
  Bearing in mind \p{hat-om} and \p{Ompsi30}, the supercharges can be rewritten in the following compact form
\begin{equation}\label{susy-char1-c}
Q= \psi^{\mu}\Pi_{\mu}\,,\qquad
Q^{{a}}= \psi^{\nu}(I^{{a}})_{\nu}{}^{\mu}\Pi_{\mu}\,,
\end{equation}
where
\begin{equation}\label{mom-c-s}
\Pi_\mu=p_\mu- {\textstyle\frac{i}{6}}\,\hat\Omega_{\mu}\,,\qquad
\hat\Omega_{\mu}=\hat\Omega_{\mu,\nu\lambda}\psi^{\nu}\psi^{\lambda}
\end{equation}
(cf. eq. (\ref{mom-c-g})).
 Using (\ref{Clifford}), (\ref{w-cov-compl-2}), (\ref{IdC}),  and (\ref{DB-8}),
one can explicitly verify the validity of  the classical ${\cal N}=4$ supersymmetry algebra
\begin{equation}\label{DB-susy}
\{Q, Q\}_{{}}=-2i H \,,\qquad \{Q^a, Q^b\}_{{}}=-2i H \delta^{ab}\,,\qquad
\{Q, Q^a\}_{{}}=0 \,,
\end{equation}
with the Hamiltonian  (\ref{ham-c-g}).

As was discussed earlier, $I^a$ obey the Clifford algebra \p{Clifford},
but not the quaternion algebra \p{quaternion}.
The tensors $I^a_{\mu\nu}$ and  $I^{\mu\nu}_a$ are obtained from \p{compl-str-free} through multiplying by
$g_{\mu\nu} = {\rm diag}(G_1 \delta_{AB}, \, G_2 \delta_{MN})$ and $g^{\mu\nu} = {\rm diag}(G_1^{-1} \delta_{AB}, \,
G_2^{-1} \delta_{MN} )$. They are all antisymmetric.

Similarly to what we had for the HKT geometry, the transformations (\ref{susy-tr1}), (\ref{susy-tr2}) can be rewritten in the form
\begin{equation}\label{susy-tr-8not}
\delta x^{\mu}= i\,\varepsilon \psi^{\mu}+i \,\varepsilon_a(I^{{a}})^{\mu}{}_{\nu}\psi^{\nu}\,,\qquad
\delta \psi^{\mu}= -\varepsilon \dot x^{\mu}+ \varepsilon_a(I^{{a}})^{\mu}{}_{\nu} \, \dot x^{\nu}\,,
\end{equation}
where $(I^{{a}})^{\mu}{}_{\nu}=g^{\mu\lambda}(I^{{a}})_{\lambda}{}^{\rho}g_{\rho\nu}$ has the same matrix components
as $(I^{{a}})_{\mu}{}^{\nu}$, i.e. $(I^{{a}})^{\mu}{}_{\nu}=(I^{{a}})_{\mu}{}^{\nu}$.

The difference with HKT is that the complex structures \p{compl-str-free} {\it are} now not covariantly constant with respect to the connections
\p{nabla8}. This means, in particular, that the torsions \p{tors1}-\p{tors3} {\it are} not given by the same expression \p{Bismut} for any $I^a$.
For sure, for every complex structure $I^a$, one can still define the Bismut connection \p{Bismut}, such that the covariant derivative
of $I^a$ vanishes. But, as opposed to the HKT case, such Bismut connections are different for different $I^a\,$                                                                              .

At the same time, the complex structures satisfy the weaker conditions
\p{w-cov-compl}
with respect to the covariant derivative \p{nabla8}.
In the considered case, these conditions amount to
\begin{equation}\label{w-cov-compl-1}
(I^a)_\lambda^{\ \rho} \Gamma_{\rho,\mu\nu}+\hat\Gamma_{\lambda,(\mu|\rho|} (I^a)_{\nu)}^{\ \rho}=0
\end{equation}
or, equivalently, to
\begin{equation}\label{w-cov-compl-2}
(I^a)_{\mu}^{\ \rho}\hat\Omega_{\rho,\nu\lambda} + (I^a)_{\nu}^{\ \rho}\hat\Omega_{\rho,\mu\lambda}=
\partial_\rho g_{\mu[\nu} (I^a)_{\lambda]}^{\ \rho} + \partial_\rho g_{\nu[\mu} (I^a)_{\lambda]}^{\ \rho}\, .
\end{equation}

The Nijenhuis concomitants \p{concom} vanish here because $(I^a)_\mu^{\ \nu}$ are constants.
One can also check that the torsion tensor satisfies the conditions \p{IdC}. Thus in the model under consideration we encounter a particular case
of the CKT geometry.

As was indicated above, for $G_1=G_2$, i.e. under the condition
\begin{equation}\label{har-G-8}
\partial_\mu\partial_\mu\,G(x,y)=0\,,
\end{equation}
the model possesses four additional  supersymmetries
 (\ref{h-susy-tr1}), (\ref{h-susy-tr2}).
In the eight-dimensional notations, these transformations can be rewritten as
\begin{equation}\label{susy-tr-8not-d}
\delta x^{\mu}= i \,\eta \tilde I^{\mu}{}_{\nu}\psi^{\nu}+i \,\eta_p(\tilde I^{{p}})^{\mu}{}_{\nu}\psi^{\nu}\,,\qquad
\delta \psi^{\mu}= \eta \tilde I^{\mu}{}_{\nu}\dot x^{\nu}+ \eta_p(\tilde I^{{p}})^{\mu}{}_{\nu}\dot x^{\nu}\,.
\end{equation}
Correspondingly, the singlet and triplet parts of the N\"other supercharges (\ref{h-susy-char}) of this ${\cal N}{=}\,4$ supersymmetry
take the form ($p=1,2,3$)
\begin{eqnarray}\label{h-susy-char1-s-2g}
\tilde Q&=& \psi^{\sigma}\tilde I_{\sigma}{}^{\mu}
\Big(p_{\mu} -{\textstyle\frac{i}{6}}\,\Omega_{\mu,\nu\lambda}\psi^{\nu}\psi^{\lambda} +
{\textstyle\frac{i}{12}}\,C_{\mu\nu\lambda}\,\psi^{\nu}\psi^{\lambda}\Big) ,
\\[6pt]
\label{h-susy-char1-tr-2g}
\tilde Q^{{p}}&=& \psi^{\sigma}(\tilde I^{{p}})_{\sigma}{}^{\mu}\Big(
p_{\mu} -{\textstyle\frac{i}{6}}\,\Omega_{\mu,\nu\lambda}\psi^{\nu}\psi^{\lambda}+{\textstyle\frac{i}{12}}\,C_{\mu\nu\lambda}\,
\psi^{\nu}\psi^{\lambda}\Big),
\end{eqnarray}
with the  complex structures $\tilde I_\mu^{\ \nu}, (\tilde I^a)_\mu^{\ \nu}$ defined in \p{ad-compl-str-free}.
These complex structures, together with the structures \p{compl-str-free}, form the Clifford algebra
\begin{equation}\label{d-Clif-compl}
\begin{array}{c}
 \{I^{a}, I^{b}\}=-2\,\delta^{ab}{\bf 1}_8\,,\qquad
\tilde I^2 = - {\bf 1}_8\,,\qquad
\{\tilde I^{p},\tilde I^{q}\}=-2\,\delta^{pq}{\bf 1}_8\,,  \\[6pt]
 \{\tilde I,\tilde I^{p}\} = \{\tilde I, I^{a}\} = \{I^a,\tilde I^{p}\} = \  0 \, .
\end{array}
\end{equation}

\subsection{${\cal N}=2$ superfield formulation}
To better understand  the geometric meaning of the spin connection \p{hat-om}, it is instructive to express the Lagrangian
\p{act-c-g} in terms of ${\cal N}=2$ superfields. In the ${\cal N}=2$ superfield language, the models studied in this paper prove to realize some special cases of
the twisted Dolbeault complexes, such that they admit extended supersymmetries.

We denote
$$
\theta^{1 1^\prime}=\eta\,,\quad \theta^{22^\prime}=-\bar\eta\,,\quad \theta^{1 2^\prime}=\theta\,,\quad \theta^{2 1^\prime}=\bar\theta\,.
$$
Then the ${\cal N}=4$ superfield  $X^{i\alpha}$ defined in (\ref{X}) can be rewritten in terms of the $SU(2)_{PG}$ doublet of chiral ${\cal N}=2$
superfields $Z^{\cal A}, \bar{Z}^{\bar {\cal A}}$:
\begin{equation}\label{X-n2}
X^{11} \equiv {\cal X}^1\,,\qquad X^{12} \equiv {\cal X}^2\,,\qquad X^{22} \equiv -\bar{\cal X}^{\bar 1}\,,
\qquad X^{21} \equiv \bar{\cal X}^{\bar 2}\,,
\end{equation}
\begin{eqnarray}\label{X-n2-nb}
{\cal X}^{\cal A}&=&Z^{\cal A}+\eta\, \epsilon^{\cal AB}\bar D\bar Z^{\bar {\,\cal B}} -i\eta\bar\eta \dot{Z}^{\cal A} \,,\\[4pt]
\bar{\cal X}^{\bar{\cal A}}&=& \bar Z^{\bar{\cal A}}-\bar \eta\, \epsilon^{{\cal \bar A}\bar{\,\cal B}} D Z^{\,\cal B}
+i\eta\bar\eta \dot{\bar Z}^{\bar{\cal A}}  \,.\label{X-n2-b}
\end{eqnarray}
Here ${\cal A}=1,2$, $\bar{\cal A}=1,2$
\footnote{The index ${\cal A}$ labeling complex chiral superfields should not be confused with the real 4-vector index $A$.}.
 Defining the  ${\cal N}=2$ superspace covariant derivatives as
\begin{equation}\label{der-n2}
D=\partial_{\theta}-i\bar \theta \partial_t\,,\qquad \bar D=-\partial_{\bar\theta}+i \theta \partial_t\,,
\end{equation}
we find that the ${\cal N}=4$ superspace constraints \p{const1} imply chirality of the ${\cal N}=2$ superfields $Z^{\cal A}$, $\bar D Z^{\cal A}=0$,
and antichirality  of $\bar Z^{\bar{\cal A}}$, $D\bar Z^{\bar{\cal A}}=0\,$.
 In their $\theta$ expansions, these superfields contain the full set of components,  $(x^{i\alpha}, \chi^{i\alpha'})$, of a
 $({\bf 4, 4, 0})$ multiplet:
\begin{eqnarray}\nonumber
Z^1&=&x^{11}+\theta\, \chi^{11^\prime} -i\theta\bar\theta\, \dot{x}^{11} \,,\nonumber\\[4pt]
Z^2&=&x^{12}+\theta\, \chi^{21^\prime} -i\theta\bar\theta\, \dot{x}^{12} \,,\nonumber\\[4pt]
\bar Z^{\bar 1}&=&-(x^{22}-\bar\theta\, \chi^{22^\prime} +i\theta\bar\theta\, \dot{x}^{22}) \,,\nonumber\\[4pt]
\bar Z^{\bar 2}&=&x^{21}-\bar\theta\, \chi^{12^\prime} +i\theta\bar\theta\, \dot{x}^{21} \,.\nonumber
\end{eqnarray}

The ${\cal N}=4$ superfield $Y^{{i}^\prime \alpha^\prime }$ (\ref{Y}) has an analogous ${\cal N}=2$ superfield decomposition:
\begin{equation}\label{Y-n2}
Y^{2^\prime2^\prime} \equiv {\cal Y}^1\,,\qquad Y^{ 1^\prime 2^\prime} \equiv {\cal Y}^{2}\,,
\qquad Y^{1^\prime1^\prime} \equiv -\bar{\cal Y}^{\bar 1}\,,\qquad Y^{ 2^\prime 1^\prime} \equiv \bar{\cal Y}^{\bar 2}\,,
\end{equation}
\begin{eqnarray}\label{Y-n2-nb}
{\cal Y}^{\cal M}&=&V^{\cal M}-\bar\eta\, \epsilon^{{\cal M}{\cal N}}\bar D\bar V^{\bar{\cal N}} +i\eta\bar\eta \dot{V}^{\cal M} \,,\\[4pt]
\bar{\cal Y}^{\bar{\cal M}}&=& \bar V^{\bar{\cal M}}+ \eta\, \epsilon^{\bar{\cal M}\bar{\cal N}} D V^{{\cal N}} -
i\eta\bar\eta \dot{\bar V}^{\bar{\cal M}}  \,,\label{Y-n2-b}
\end{eqnarray}
where  ${\cal M}=1,2$, $\bar{\cal M}=1,2\,$. The chiral superfields $V^{\cal M}$, $\bar D V^{\cal M}=0$, and the antichiral superfields
$\bar V^{\bar{\cal M}}$, $D\bar V^{\bar{\cal M}}=0$, are defined by
\begin{eqnarray}\nonumber
V^1&=&y^{2^\prime2^\prime}+\theta\, \psi^{22^\prime} -i\theta\bar\theta\, \dot{y}^{2^\prime2^\prime} \,,\nonumber\\[4pt]
V^2&=&y^{2^\prime 1^\prime }+\theta\, \psi^{21^\prime} -i\theta\bar\theta\, \dot{y}^{ 2^\prime 1^\prime} \,,\nonumber\\[4pt]
\bar V^{\bar 1}&=&-(y^{1^\prime1^\prime}-\bar\theta\, \psi^{11^\prime} +i\theta\bar\theta\, \dot{y}^{1^\prime1^\prime}) \,,\nonumber\\[4pt]
\bar V^{\bar 2}&=&y^{1^\prime 2^\prime}-\bar\theta\, \psi^{12^\prime} +i\theta\bar\theta\, \dot{y}^{1^\prime 2^\prime } \,.\nonumber
\end{eqnarray}
Their chirality and antichirality follow from the ${\cal N}=4$ superspace constraint~(\ref{const2}).
Note that the complex coordinates $Z^{\cal A}, V^{\cal M}$ appear as the components of the complex eight-vector
$x_\mu + i(I^3)_\mu^{\ \nu} x_\nu$,  with $I^3 = I^{a=3}$ and $I^a$ being defined in  \p{compl-str-free}.

The ${\cal N}=4$ superfield action (\ref{act}) can be equivalently rewritten in the ${\cal N}=2$ superfield form as
\be
S = \int dt d\theta d\bar\theta {\cal L}(Z^{\cal A}, \bar{Z}^{\bar{\cal A}}, V^{\cal M}, \bar{V}^{\bar{\cal M}})\,,\label{act-c-n2}
\ee
with
\begin{eqnarray}
{\cal L}(Z^{\cal A}, \bar{Z}^{\bar{\cal A}}, V^{\cal M}, \bar{V}^{\bar{\cal M}}) &=&
\int d\eta d\bar\eta\, {\cal L}({\cal X}^{\cal A},\bar{\cal X}^{\bar{\cal A}},{\cal Y}^{\cal M}, \bar{\cal Y}^{\bar{\cal M}}) \nn
&=& \!\quad
\Big[(\partial_{\cal A}\partial_{\bar{\cal B}}{\cal L})+
\epsilon_{{\cal A}{\cal C}}\,\epsilon_{\bar{\,\cal B}\bar{\cal D}}\,
(\partial_{\cal D}\partial_{\bar{\cal C}}{\cal L}) \Big]D Z^{\cal A} \bar D\bar Z^{\bar{\,\cal B}} \nn
&& -\Big[(\partial_{\cal M}\partial_{\bar{\cal N}}{\cal L})+
\epsilon_{{\cal M}{\cal K}}\,\epsilon_{\bar{\cal N}\bar{\cal L}}\,(\partial_{\cal L}\partial_{\bar{\cal K}}{\cal L}) \Big]
D V^{\cal M} \bar D\bar V^{\bar{\cal N}}\nonumber \\[4pt]
&& +\,\epsilon_{{\cal A}{\cal B}}\,\epsilon_{{\cal M}{\cal N}}\,(\partial_{\bar{\cal B}}\partial_{\bar{\cal N}}{\cal L})
\,D Z^{\cal A} D V^{\cal M}\nonumber \\[4pt]
&& -\,\epsilon_{\bar{\cal A}\bar{\cal B}}\,\epsilon_{\bar{\cal M}\bar{\cal N}}\,(\partial_{\cal B}\partial_{\cal N}{\cal L})
\,\bar D \bar Z^{\bar{\cal A}} \bar D\bar V^{\bar{\cal M}}\,.\label{act-c-n23}
\end{eqnarray}
 The action \p{act-c-n2} with the Lagrangian \p{act-c-n23} is a particular case of the ${\cal N}=2$ superfield action
for CKT systems given in~\cite{DI}.

Note an important difference between the ${\cal N}=2$ superfield expansions of
${\cal X}^{\cal A}$ and ${\cal Y}^{\cal M}$ \p{X-n2-nb} and \p{Y-n2-nb}:
in the second case the Grassmann variable $\eta$ is replaced by $\bar\eta$. Clearly, this difference is not
important, when considering the systems with the single $X^{i\alpha}$ or the single $Y^{i' \alpha' }$, since one can always redefine
$\eta \leftrightarrow \bar\eta$.  
It becomes, however, essential while considering
these supermultiplets together, because it gives rise to different transformation laws of the corresponding chiral superfields under
the hidden ${\cal N}=2$ supersymmetry acting as shifts of $\eta, \bar\eta$.
 From the ${\cal N}=2$ superfield point of view,
it is exactly this difference that is responsible for the emergence of the CKT
geometry in such a system  \cite{DI}.

Indeed, the presence of $\eta$ in \p{X-n2-nb} and $\bar\eta$ in \p{Y-n2-nb} leads to the appearance of the
mixed structures
$D Z^{\cal A} D V^{\cal M}$, $\bar D \bar Z^{\bar{\cal A}} \bar D\bar V^{\bar{\cal M}}$ (the last two terms in \p{act-c-n23}).
They  have the same form as \p{holtorsion} and thus yield  extra {\it holomorphic} torsion components. Such mixed terms involving
the same type of ${\cal N}=2$ spinor derivatives are absent in the ${\cal N}=2$ superfield formulation of the HKT models: only
the structures similar to the first two terms in \p{act-c-n23} appear there \cite{Hull,DI}. These structures produce the target space metric and
the Bismut-type torsions having no (anti)holomorphic components.

We can now understand the crucial difference between the CKT and the HKT models. Being expressed through ${\cal N} = 2$ superfields,
the latter represent special types of the usual  Dolbeault models \p{Dolb}, \p{WZ}, but without  \p{holtorsion}. At the same time,
the CKT models are characterized by the inevitable presence of the terms of type \p{holtorsion} in the corresponding ${\cal N} = 2$ superfield actions.

In the operator language, the models of this class are obtained by a similarity transformation of the complex supercharges,
 \be
\lb{simil}
Q \ \to \ e^{{\cal B}_{jk} \psi^j \psi^k} Q e^{{-\cal B}_{jk} \psi^j \psi^k}\, .
 \ee
In the CKT models, the fermion charge $\psi_j \bar \psi^j$ {\it is not conserved}! This is another difference from the HKT models
where the fermion charge {\it is} conserved. As there are three different complex structures and three ways to introduce complex coordinates
(and, hence, three ways to define the chiral superfields $Z^{\cal A}$), there are three such conserved charges $F^a = (I^a)_{\mu\nu} \psi^\mu \psi^\nu $.
Then the triplet HKT supercharges are obtained from the singlet one through the commutation relation \cite{Verb,QHKT},
  \be
\lb{commFQ}
S^a = [Q, F^a] \, .
 \ee
In the interacting CKT case, there are no conserved fermion charges and, therefore, no relation like \p{commFQ} can be written.

On the Lagrangian level, the same phenomenon manifests itself as the invariance of the general HKT actions under one of the
automorphism $SU(2)$ symmetries, that one which is realized solely on fermionic fields, and the non-invariance of the general CKT actions under
{\it any} of the  $SU(2)$ symmetries involved. These specific features were already mentioned in the previous Sections.

\setcounter{equation}0
\section{Quantum supercharges and geometry}
We construct first the quantum supercharges for the CKT model with the metric \p{g8}. The experience thus acquired will allow us
to find quantum supercharges for a generic CKT(OKT) manifold and, based on this, to suggest new definitions of CKT and OKT geometries.

A general recipe to quantize supersymmetric theories was given in \cite{howto}. The correct quantum {\it supercharges}
(not the Hamiltonian !) are obtained by ordering the classical supercharges according to symmetric Weyl prescription.
Obviously, Weyl ordering of real classical expressions gives Hermitian operators, $Q^\dagger = Q$. These operators
act on the wave functions normalized with the flat measure. In our case,
  \be
\label{normflat}
\int \prod_\mu dx^\mu  \prod_a  d \psi^a  d \bar \psi^a  \,
\exp\{\bar \psi^a \psi^a\} \, \bar \Psi(x^\mu,
\bar \psi^a)
 \Psi(x^\mu, \psi^a) \ =\ 1\, ,
 \ee
where the pair $(\psi^a, \bar \psi^a)$ corresponds to  some particular splitting of the flat $(\psi^{\underline{\mu}})$
into the canonical coordinates and  momenta. The choice of such a splitting is quite arbitrary. It does not affect the results.

For SQM sigma models describing the motion over curved manifold, it is more convenient to consider
{\it covariant} wave functions normalized
 with the extra factor $\sqrt{g}$. The Hermitian with respect to this Riemannian measure
operators ${\cal O}$  are obtained from the ``flat''
Weyl-ordered operators by the similarity transformation,
 \be
\lb{simtransf}
 {\cal O}^{\rm cov} \ =\ g^{-1/4} {\cal O}^{\rm flat} g^{1/4} \, .
 \ee

Applying this recipe to the singlet supercharge \p{susy-char1-s-2g}, one reproduces the
{\it same} expression as \p{susy-char1-s-2g} with the same order of quantum operators.

It is not so straightforward to derive the expression for the triplet quantum supercharge, but it can be done using the results of Sect.
4.4. Consider a  pair $(Q, Q^{\hat a})$ with some particular $a=\hat a$. As we have seen, in the ${\cal N} = 2$ superfield formulation
our model amounts to the twisted Dolbeault complex with extra holomorphic and antiholomorphic torsion components.
These (anti)holomorphic components are given by the tensor $H_{\mu\nu\lambda}(C,I)$ defined in
 \p{holdobavka}, with $I \equiv I^{\hat a}$. Now we can understand the origin of the coefficient $1/4$ in  \p{holdobavka}. Just with this coefficient,
 the property $H_{\mu\nu\lambda}(H,I) = H_{\mu\nu\lambda}$ holds. Thus the total torsion $C_{\mu\nu\lambda}$ is expressed as a sum of
 its (anti)holomorphic part $H_{\mu\nu\lambda}$ and a part $B_{\mu\nu\lambda}$ satisfying the condition \p{holproj} which strips off
 all (anti)holomorphic components from $B_{\mu\nu\lambda}\,$.

 As was shown in \cite{FIS}, the full  quantum covariant ${\cal N}=2$ supercharges
(derived by the prescription explained above) are the sum of ``untwisted'' quantum supercharges (i.e., without holomorphic torsions)
and the holomorphic and antiholomorphic
parts $\sim C_{jkl}, \sim C_{\bar j\bar k \bar l}$. In our real vector notation, the latter correspond just to the term
$\frac i{12} H_{\mu\nu\lambda} \psi^\mu \psi^\nu \psi^\lambda $ for the supercharge $Q$ and the term $\frac i{12} H_{\mu\nu\lambda}
(I^{\hat a})_\sigma^{\ \mu} \psi^\sigma \psi^\nu \psi^\lambda $ for the supercharge $Q^{\hat a}$ (with the fixed value of the index $a={\hat a}$),
which are {\it the same} as in the classical supercharges \p{susy-char1-s-2g}, \p{susy-char1-tr-2g}
   \footnote{This follows from our Weyl ordering prescription and from the total antisymmetry of $H_{\mu\nu\lambda}$ and of the tensor
   \be
\lb{G=IH}
  G_{\mu\nu\lambda} =
  (I^{\hat a})_\mu^{\ \alpha}   H_{\alpha\nu\lambda} \ =\ \frac 14 \left[ (I^{\hat a})_\mu^{\ \alpha}   C_{\alpha\nu\lambda}
 + (I^{\hat a})_\nu^{\ \alpha}   C_{\mu \alpha \lambda} + (I^{\hat a})_\lambda^{\ \alpha}   C_{\mu \nu \alpha} -
  (I^{\hat a})_\mu^{\ \alpha}   (I^{\hat a})_\nu^{\ \beta} (I^{\hat a})_\lambda^{\ \gamma}    C_{\alpha \beta \gamma} \right].
   \ee
}.
 The untwisted part of $Q^{\hat a}$ involves the Bismut torsion tensor $B^{\hat a}_{\mu\nu\lambda}$ for the given complex structure $I^{\hat a}$, and it is
represented by the same expression as in \p{susy-char1-tr-2f} with the  order of indices indicated there \cite{QHKT}.
Finally, we obtain the following expressions for the quantum supercharges,
\begin{eqnarray}
\label{Qqusing}
Q&=& \psi^{\mu}
\Big(p_{\mu} -{\textstyle\frac{i}{2}}\,\Omega_{\mu,\nu\lambda}\psi^{\nu}\psi^{\lambda} +
{\textstyle\frac{i}{12}}\,C_{\mu\nu\lambda}\,\psi^{\nu}\psi^{\lambda}\Big) ,
\\[6pt]
\label{Qqutrip}
Q^{{\hat a}}&=& \psi^{\sigma}(I^{{\hat a}})_{\sigma}{}^{\mu}\Big(
p_{\mu} -{\textstyle\frac{i}{2}}\,\Omega_{\mu,\nu\lambda}\psi^{\nu}\psi^{\lambda}
-{\textstyle\frac{i}{4}}\,B^{\hat a}_{\mu\nu\lambda} \,\psi^{\nu}\psi^{\lambda}
+ {\textstyle\frac{i}{12}}\,H^{\hat a}_{\mu\nu\lambda} \,\psi^{\nu}\psi^{\lambda}\Big),
\end{eqnarray}
 with
  \be
  \lb{C=B+H}
C_{\mu\nu\lambda} \ =\ B^{\hat a}_{\mu\nu\lambda} + H^{\hat a}_{\mu\nu\lambda}\, .
  \ee
$B^{\hat a}_{\mu\nu\lambda}$ is the Bismut torsion for the complex structure $I^{\hat a}$ and $H^{\hat a}_{\mu\nu\lambda}(C, I^{\hat a})$ was defined in
\p{holdobavka}.

Note that the {\it classical} triplet supercharge can also be represented as \p{Qqutrip}. This geometric representation should be valid
for {\it any} CKT manifold (the more compact representation \p{susy-char1-tr-2g} probably holds only for our particular model with the metric \p{g8}).

It remains to prove that the supercharges \p{Qqusing}, \p{Qqutrip} indeed form the quantum ${\cal N}=4$ superalgebra.
The  ${\cal N}=2$ superalgebra relations for a fixed index $a={\hat a}$,
 $$Q^2 = (Q^{\hat a})^2 = H, \ \ \ \ \ \ \ \ \ \ \ \{Q, Q^{\hat a} \}_+ = 0\,, $$
follow from the analysis in Ref. \cite{FIS}. It remains to show that $\{Q^{\hat a}, Q^{\hat b}\}_+ = 0$ when ${\hat a} \neq {\hat b}$.

We will consider first our specific CKT model with the metric \p{g8} and the classical supercharges \p{susy-char1-tr-2g}.
The best  way to proceed is
to capitalize  on the fact that the quantum supercharges \p{Qqutrip} are obtained from the classical supercharges
\p{susy-char1-tr-2g} by Weyl ordering and on the well-known assertion that the Weyl symbol of an
(anti)commutator of two operators is given by the
Gr\"onewold-Moyal (G-M) bracket of their Weyl symbols \cite{Moyal}.
For a system involving fermion variables, the latter is defined as
\begin{eqnarray}
 i\hbar \{A, B\}_{GM} \!\!\!\!\!&=&\!\!\!\!\! 2\sinh \left\{ \frac \hbar 2 \sum_a \left( \frac {\partial^2}{\partial \psi^{(2)}_a
\partial \bar \psi^{(1)}_a } -
\frac {\partial^2}{\partial \psi^{(1)}_a \partial \bar \psi^{(2)}_a }  \right)  +
\frac {i\hbar}2 \sum_i \left( \frac {\partial^2}{\partial q^{(1)}_i \partial p^{(2)}_i } -
\frac {\partial^2}{\partial q^{(2)}_i \partial p^{(1)}_i }    \right) \right\} \nonumber\\
&&
\left. A\left(p_i^{(1)}, q_i^{(1)}; \bar \psi_a^{(1)}, \psi^{(1)}_a \right)
B\left(p_i^{(2)}, q_i^{(2)}; \bar \psi_a^{(2)}, \psi^{(2)}_a \right) \right|_{1=2} \, ,
 \lb{Moyal}
\end{eqnarray}
where $(p_i,\, q_i)$ are bosonic and $(\psi^a, \,\psi^a)$ --- fermionic canonically conjugate pairs. We introduced here
$\hbar$ to make more explicit the classical limit $\hbar \to 0$. In this limit, only the first term in the expansion of {\it sinh} actually
contributes, and the G-M bracket is reduced to the Poisson bracket. As we have
mentioned before, the Poisson
bracket $\{Q^a, Q^b\}_{P.B.}$ vanishes for $a \neq b$.

For the supercharges \p{susy-char1-tr-2g}, also the cubic term in the expansion of $\sinh$ with six partial derivatives
might contribute. There are two such possible contributions: the contribution with the derivatives of vielbein
 from the first terms in  \p{susy-char1-tr-2g}
and the contribution $\sim \partial^6/(\partial \psi)^6$ from the second terms. The first contribution is
\be
\lb{derivviel}
\sim (\partial_\alpha e^\sigma_a) I_\sigma^{\ \nu} (\partial_\nu e_a^{\ \mu}) J_\mu^{\ \alpha} \, ,
  \ee
where $I,J$ are two different complex structures \p{compl-str-free}. It is not difficult to see that, for the diagonal
metric \p{g8} and the naturally chosen diagonal vielbeins, \p{derivviel} vanishes.

The second possible contribution has the form
 \be
\lb{deriv6psi}
 \sim \hat \Omega_{\mu, \nu\lambda} \hat \Omega_{\alpha, \gamma\delta} I_\sigma^{\ \mu} J_\rho^{\ \alpha}
\left(g^{\sigma\rho} g^{\nu\gamma} g^{\lambda\delta} + 2 g^{\sigma\gamma} g^{\nu\delta} g^{\lambda\rho} \right).
   \ee
This also vanishes, as can be checked using the explicit expressions \p{tors1}-\p{tors3}, \p{Om1}-\p{Om3}.

\subsection{Generic CKT/OKT geometry. Theorems and definitions}
The calculation of $\{Q^a, Q^b\}_{G-M}$ given above was performed for the particular model \p{g8}.
But one can show that it is proportional to $\delta^{ab}$ also in the most general case.
We are going to prove the following

\begin{thm}
 Let $I^a$ be three (seven) integrable complex structures with the vanishing  Nijenhuis concomitants \p{concom}
that satisfy the Clifford algebra \p{Clifford}. Let $C_{\mu\nu\lambda}$ be a totally antisymmetric torsion tensor representable as
in \p{C=B+H} for each complex structure. Let each pair of the supercharges \p{Qqusing}, \p{Qqutrip} satisfy the minimal
 ${\cal N} = 2$ superalgebra. Then they all together satisfy
 the extended ${\cal N} = 4$ (   ${\cal N} = 8$)
supersymmetry algebra, and we are dealing with a CKT (OKT) geometry.
\end{thm}

\begin{proof}
 The definition of the CKT/OKT geometries given in the Introduction involves the conditions \p{w-cov-compl}, \p{eq-cs1a} and also
\p{concom}. The latter is also among the conditions we are imposing now.
It is this condition involving a pair of complex structures which guarantees the property $\{Q^a, Q^b\} \propto \delta^{ab}$.
As for the conditions \p{w-cov-compl}, \p{eq-cs1a}, they are formulated for each complex structure separately and should be
there once we require the fulfillment of the  ${\cal N} = 2$ algebra for each pair $(Q, \,Q^{\hat a})$.
Let us check this explicitly.

$\bullet$ The property \p{w-cov-compl} follows from the decomposition \p{C=B+H}, the fact that the complex structures
are covariantly constant with respect to their Bismut connections and the total antisymmetry of $H^{\hat a}_{\mu\nu\lambda}$ in its lower case indices.

$\bullet$ To prove  \p{eq-cs1a}, we introduce, for each complex structure $I^{\hat a}$, the corresponding complex coordinates $z^{j({\hat a})}$ and
$\bar z^{\bar j ({\hat a})}$ (such that the metric tensor is Hermitian, $g_{\mu\nu} \to h_{j \bar k}$) and
then represent the 3-form $C \ =\ C_{\mu\nu\lambda} dx^\mu \wedge dx^\nu \wedge dx^\lambda$ as a sum of four terms,
 \be
C \ =\ C_{3,0} + C_{2,1} + C_{1,2} + C_{0,3} \, ,
 \ee
where $C_{p,q}$ involves the components with $p$ holomorphic and $q$ antiholomorphic indices. We  can also represent $d = \partial^{\hat a} +
\bar\partial^{\hat a}$, where $\partial^{\hat a}, \, \bar\partial^{\hat a}$ are exterior derivatives holomorphic (antiholomorphic) with respect to
the complex structure $I^{\hat a}$ \footnote{Hereafter, for better readability, the index ${\hat a}$ on $C_{p,q}$,  etc,
will be omitted.}.

The components $C_{2,1}$ and   $C_{1,2}$ come from the Bismut torsion $B_{\mu\nu\lambda}$. They are
\footnote{ See, e.g., Eq.(2.5) in Ref. \cite{IS}.}
 $  C_{2,1} \propto \partial \omega $, \  $  C_{1,2} \propto \bar \partial \omega $,  where
$\omega = h_{j\bar k} dz^j \wedge d\bar z^{\bar k}$. Hence $\partial C_{2,1} = \bar \partial C_{1,2} = 0$.
The components $C_{3,0}$, $C_{0,3}$ come from $H^{\hat a}_{\mu\nu\lambda}\,$. They are also $\partial$ (resp. $\bar\partial$) exact,
$C_{3,0} \propto \partial {\cal B}$,  $C_{0,3} \propto \bar \partial \bar {\cal B}$ \cite{FIS}.

Then the l.h.s. of  Eq. \p{eq-cs1a} is
 \be
\lb{iota-d}
\iota dC \ =\ \iota \left( \bar\partial C_{3,0} + \bar\partial C_{2,1} + \partial C_{1,2} +  \partial C_{0,3} \right)
 = 2i \left(\bar\partial C_{3,0} - \partial C_{0,3} \right).
 \ee

The r.h.s.  of  Eq. \p{eq-cs1a} is
 \be
\lb{d-iota}
 \frac {2i}3 (\partial + \bar\partial) (3C_{3,0} + C_{2,1} - C_{1,2} - 3C_{0,3} ) =
2i \left(\bar\partial C_{3,0} - \partial C_{0,3} \right) + \frac {2i}3 \left(\bar \partial C_{2,1} - \partial C_{1,2} \right).
 \ee

The last term in \p{d-iota} vanishes, $ \bar\partial \partial \omega -  \bar\partial \partial \omega = 0\,$. Thus the expressions
\p{iota-d} and \p{d-iota} coincide. This proves the validity of the condition    \p{eq-cs1a} and, hence,  our theorem.
\end{proof}

The conditions of the Theorem just proven can thus serve as definitions of CKT (OKT) manifold which look as a natural
generalization of the definition of a HKT manifold given in the Introduction. An HKT manifold is a manifold with 3 quaternionic
complex structures for which Bismut connections coincide. A CKT (OKT) manifold is a manifold with 3 (7) complex structures for
which Bismut connections as such do not coincide, but the total connections involving the torsions \p{C=B+H}, with the Bismut parts
``completed'' by extra components which are $\partial$--exact holomorphic and $\bar\partial$--exact antiholomorphic
with respect to each complex structure,  {\it do} coincide. This composite total connection satisfies basic conditions
\p{w-cov-compl}, \p{concom} and \p{eq-cs1a} of the CKT (OKT) geometry.

\setcounter{equation}0
\section{Examples}

Let us consider a few particular cases of the model (\ref{act-c-eq}), (\ref{act-c-g}).

\subsection{$S^4\times S^4$}

In this case,
$$
G_1=\frac{\alpha_1}{(1+\beta_1 x^A x^A)^2}\,,\qquad  G_2=\frac{\alpha_2}{(1+\beta_2 y^M y^M)^2}\,,
$$
where $\alpha_{1,2}$, $\beta_{1,2}$ are constants.
The metric is produced by the superfield Lagrangian
$$
{\cal L}=\frac{\alpha_1}{4(\beta_1)^2X^AX^A}\Big[1+\ln(1+\beta_1X^BX^B)\Big]
-\frac{\alpha_2}{4(\beta_2)^2Y^MY^M}\Big[1+\ln(1+\beta_2Y^KY^K)\Big].
$$
In this case all geometric characteristics (\ref{Gamma1})-(\ref{Gamma3}), (\ref{tors1})-(\ref{tors3}), (\ref{Om1})-(\ref{Om3})
are sums of terms originating from the two independent sectors.
As a result, the system possesses two types of  ${\cal N}=\,4$ supersymmetries.
Besides the singlet supersymmetry generator
\begin{eqnarray}
Q&=&Q_{(1)}+Q_{(2)}\,,\\
Q_{(1)}&=&\chi^{A}\,\Big[p_{A} - {\textstyle\frac{\alpha_1\beta_1 i}{3(1+\beta_1 x\cdot x)^3}}\,\epsilon_{ABCD}\,\chi^{B}\chi^{C}x^D\Big],\\
Q_{(2)}&=&\psi^{M}\,\Big[p_{M} - {\textstyle\frac{\alpha_2\beta_2 i}{3(1+\beta_2 y\cdot y)^3}}\,\epsilon_{MNKL}\,\psi^{N}\psi^{K}y^L\Big],
\end{eqnarray}
one can define two different triplet supersymmetry generators,
\begin{equation}\label{S-44}
S^a=Q_{(1)}^a+Q_{(2)}^a
\end{equation}
and
\begin{equation}\label{Q-44}
Q^a=Q_{(1)}^a-Q_{(2)}^a\,,
\end{equation}
where
\begin{eqnarray}
Q_{(1)}^a&=&-\chi^{A}\eta^a_{AB}p_{B} - {\textstyle\frac{\alpha_1\beta_1 i}{(1+\beta_1 x\cdot x)^3}}\,\chi^{A}\eta^a_{AB}\chi^{B}\,\chi^{C}x^C\,,\\
Q_{(2)}^a&=&-\psi^{M}\eta^a_{MN}p_{N} - {\textstyle\frac{\alpha_2\beta_2 i}{(1+\beta_1 y\cdot y)^3}}\,\psi^{M}\eta^a_{MN}\psi^{N}\,\psi^{K}y^K\,.
\end{eqnarray}
The generators (\ref{S-44}) correspond to the HKT geometry whereas (\ref{Q-44}) to the CKT geometry.

Since two sectors completely decouple, this system is quite similar to that of
two free $({\bf 4, 4, 0})$ multiplets considered earlier.
To gain a less trivial example, it is necessary to allow $G_1$ to depend on $y^M$ and/or $G_2$ to depend on $x^A$.

\subsection{Non-trivial monopole-like case}

Let us consider symmetric case with the Lagrangian
\begin{equation}
{\cal L}=-\frac{X^AX^A}{8}\left(\alpha_1+\frac{\beta_1}{Y^MY^M}\right)
+\frac{Y^MY^M}{8}\left(\alpha_2+\frac{\beta_2}{X^AX^A}\right),
\end{equation}
where $\alpha_{1,2}$, $\beta_{1,2}$ are constants. Then the metric functions are
\begin{equation}
G_1=\alpha_1+\frac{\beta_1}{y^My^M}\,,\qquad \qquad G_2=\alpha_2+\frac{\beta_2}{x^Ax^A}\,.
\end{equation}
The ${\cal N}=\,4$ supersymmetry generators read
\begin{equation}
Q=\chi^{A}p_{A} +\psi^{M}p_{M} +{\textstyle\frac{i}{4}}\,C_{ABM}\chi^{A}\chi^{B}\psi^{M}+{\textstyle\frac{i}{4}}\,C_{MNA}\psi^{M}\psi^{N}\chi^{A}\,,
\end{equation}
where
\begin{equation}
C_{ABM}=\left(\epsilon_{MABK}+
2\delta_{M[A}\delta_{B]K} \right)\frac{\beta_1y^K}{(y\cdot y)^{2}}\,,
\qquad
C_{MNA}=\left(\epsilon_{AMNB}+
2\delta_{A[M}\delta_{N]B} \right)\frac{\beta_2x^A}{(x\cdot x)^{2}}\,,
\end{equation}
and
\begin{eqnarray}
Q^a&=&-\chi^{A}\eta^a_{AB}p_{B} +\psi^{M}\eta^a_{MN}p_{N} \\
&& \nonumber -\,
{\textstyle\frac{i}{6}}\,\eta^a_{AC}(C_{CBM}-2\Omega_{C,BM})\chi^{A}\chi^{B}\psi^{M}
+{\textstyle\frac{i}{6}}\,\eta^a_{MK}(C_{KNA}-2\Omega_{K,NA})\psi^{M}\psi^{N}\chi^{A}\,,
\end{eqnarray}
where
\begin{equation}
C_{ABM}=-\delta_{AB} \frac{\beta_1y^M}{2(y\cdot y)^{2}}\,,
\qquad
C_{MNA}=-\delta_{MN} \frac{\beta_2x^A}{2(x\cdot x)^{2}}\,.
\end{equation}

This case is non-trivial as there is a mixing of fermionic variables from different $({\bf 4,4,0})$ multiplets in the supercharges.

\subsection{Conformally flat OKT manifold}

As the last example,  we consider the system with ${\cal N}=\,8$ supersymmetry and OKT geometry.
The metric should obey the condition $G_1=G_2=G\,$, with $G$ satisfying  the eight-dimensional harmonicity equation
$\partial_\mu\partial_\mu G=0$. The solution with the maximal $O(8)$ symmetry is the conformally flat metric with \cite{ILS}
\begin{equation}
G_1=G_2=\alpha+\frac{\beta}{(x^\mu x^\mu)^3}\,,
\end{equation}
where $\alpha$, $\beta$ are constants. This metric can be derived from the superfield Lagrangian
\begin{equation}
{\cal L}=\frac{\alpha}{8}\left(Y^MY^M-X^AX^A\right)+\frac{\beta}{16X^\mu X^\mu}\left(\frac{1}{Y^MY^M}-\frac{1}{X^AX^A}\right).
\end{equation}
Substituting
$$
\partial_\mu G=-6\beta x^\mu (x^\nu x^\nu)^{-4}
$$
into (\ref{tors1})-(\ref{tors3}), (\ref{Om1})-(\ref{Om3}) and then into (\ref{susy-char1-s-2g})-(\ref{susy-char1-tr-2g}),
(\ref{h-susy-char1-s-2g})-(\ref{h-susy-char1-tr-2g}), we obtain the explicit expressions for the generators of
${\cal N}=\,8$ supersymmetry.

\setcounter{equation}{0}
\section{Summary and outlook}

 The bulk of this paper was devoted to a specific ${\cal N}=\,4$ SQM
 system  formed by the two interacting ({\bf 4, 4, 0}) supermultiplets.
We have shown that in the case when both supermultiplets are of the same nature, this system
describes the HKT Dolbeault complex. On the other hand, the system with two mutually mirror multiplets
corresponds to  CKT geometry. In a particular case with conformally flat harmonic 8-dimensional metric,
the ${\cal N}=4$ supersymmetry is enhanced to  ${\cal N}=\,8$,  yielding OKT geometry.

We have found the
explicit expressions for the classical and quantum supercharges, and this allowed us to make
certain observations concerning the mathematical structure of {\it generic} CKT and OKT complexes.
Their new transparent
definitions were given in Sect. 5. As was explained there,
the inherent feature of CKT and OKT geometries is the presence of holomorphic torsions.

We considered the system with only two (${\bf 4, 4, 0}$) supermultiplets.
Including more ${\cal N}=\,4$ supermultiplets, we would meet a richer situation.
For example, let us consider four  (${\bf 4, 4, 0}$) supermultiplets. When we consider them separately,
the  supercharges associated with each multiplet are ($Q_{(1)}$, $Q_{(1)}^a$); ($Q_{(2)}$, $Q_{(2)}^a$); ($Q_{(3)}$, $Q_{(3)}^a$); ($Q_{(4)}$, $Q_{(4)}^a$).
The interaction of these multiplets can pick up one set of the
supersymmetry generators from the following variants:
\begin{equation}
S^a=Q_{(1)}^a+Q_{(2)}^a+Q_{(3)}^a+Q_{(4)}^a\,;
\end{equation}
\begin{equation}
Q^a=Q_{(1)}^a+Q_{(2)}^a+Q_{(3)}^a-Q_{(4)}^a\,;
\end{equation}
\begin{equation}
\hat Q^a=Q_{(1)}^a+Q_{(2)}^a-Q_{(3)}^a-Q_{(4)}^a\,.
\end{equation}
Supersymmetry with $S^a$ corresponds to the HKT geometry, whereas both
$Q^a$ and $\hat Q^a$ are pertinent to the CKT geometry. Probably, the supercharges $S^a$ arise when all four
multiplets are of the same type, the supercharges $Q^a$ are relevant to a system of three ordinary and one mirror multiplets and
the supercharges  $\hat Q^a$ appear when one is playing with two ordinary and two mirror multiplets.
For some particular class of metrics, the latter system is actually OKT enjoying extended
${\cal N}=\,8$ supersymmetry.
It would be interesting to check all this by explicit calculations.

For HK \cite{Wipf} and HKT \cite{QHKT} systems,  ${\cal N}=\,4$ supersymmetry
is kept intact when gauge fields
of some particular form are added.
To this end, the field strength tensors $F_{\mu\nu}$
should commute with all complex structures. This is true, e.g.,
for a $4D$ (anti)self-dual field \cite{SKon}. It would be interesting to study the possibility of incorporating gauge fields for the CKT and OKT systems.
As was mentioned in the Introduction, in the  ${\cal N}=\,2$ superfield language adding gauge fields  corresponds  to
adding  the Wess-Zumino terms \p{WZ}  to the action.
In the ${\cal N}=\,4$ superfield framework, it  is possible \cite{IL,ISKon,DI} within
the harmonic superspace description.
The simultaneous inclusion of the mutually
mirror ${\cal N}=\,4$ supermultiplets would require the bi-harmonic superspace approach \cite{IN}.

It is also an open question which geometries are associated with the
${\cal N}=\,4$ and ${\cal N}=\,8$ SQM
systems built on the basis of those ${\cal N}=\,4$ supermultiplets (including the mirror ones) \cite{BIKL,ABC,Franc,ILS,DIN8},
which are different from the ({\bf 4, 4, 0}) multiplets exploited here.

\bigskip
\section*{Acknowledgements}

\noindent
S.F. \&\ E.I. acknowledge support from the RFBR
grants 12-02-00517, 13-02-90430 and
a grant of the IN2P3-JINR Programme. They would like to thank SUBATECH,
Universit\'{e} de Nantes, for the warm hospitality
in the course of this study.
The work of E.I. was partially carried out under the Convention N${}^{\rm o}$ 2010 11780.

\bigskip

\renewcommand\theequation{A.\arabic{equation}} \setcounter{equation}0
\section*{Appendix A.\quad Some properties of the `t Hooft symbols}

In this Appendix we present the definition of the `t Hooft symbols and the related identities (for details, see \cite{tH,BVN}).

The ${\rm SO}(4)$ sigma matrices $\sigma_A$, $A=1,2,3,4$ can be chosen in the form
\begin{equation}\label{4sigma}
\sigma_A=(\sigma_1,\sigma_2,\sigma_3;\sigma_4)=(\vec{\sigma};i)\,,\qquad \sigma^\dag_A=(\vec{\sigma};-i)\,,
\end{equation}
where $\vec{\sigma}$ are ordinary Pauli matrices. The matrices (\ref{4sigma}) satisfy the identities
\begin{equation}\label{id-4sigma}
\sigma^\dag_A\sigma_B=\delta_{AB}+i\,\eta^a_{AB}\,\sigma_a\,,\qquad \sigma_A\sigma^\dag_B=\delta_{AB}+i\,\bar\eta^a_{AB}\,\sigma_a\,,
\end{equation}
where $\sigma_a$, $a=1,2,3$ are Pauli matrices and
\begin{equation}\label{thooft}
\eta^a_{BC}=-\eta^a_{CB}=\left\{ \begin{array}{ll}
                      \epsilon_{aBC}&\quad B,C=1,2,3\,,\\
                      \delta_{aB}&\quad C=4 \,,
                    \end{array}
\right.
\end{equation}
\begin{equation}\label{b-thooft}
\bar\eta^a_{BC}=-\bar\eta^a_{CB}=\left\{ \begin{array}{ll}
                      \epsilon_{aBC}&\quad B,C=1,2,3\,,\\
                      -\delta_{aB}&\quad C=4\,,
                    \end{array}
\right.
\end{equation}
 are the `t Hooft symbols. They satisfy the following identities
\begin{eqnarray}\label{hooft1}
\eta^{{a}}_{AB}\eta^{{b}}_{AC}&=& \delta^{{a}{b}}\delta_{BC}+\epsilon^{{a}{b}{c}}\eta^{{c}}_{BC}\,, \\[4pt]
\eta^{{a}}_{AB}\eta^{{a}}_{CD}&=& \delta_{AC}\delta_{BD}-\delta_{AD}\delta_{BC}+\epsilon_{ABCD}\,, \label{hooft2}\\[4pt]
\epsilon^{abc}\eta^{{b}}_{AB}\eta^{{c}}_{CD}&=& \delta_{AC}\eta^{{a}}_{BD}+\delta_{BD}\eta^{{a}}_{AC}-
\delta_{AD}\eta^{{a}}_{BC}-\delta_{BC}\eta^{{a}}_{AD}\,, \label{hooft3}\\[4pt]
\epsilon_{BCDE}\eta^{{a}}_{AE}&=& \delta_{AB}\eta^{{a}}_{CD}-\delta_{AC}\eta^{{a}}_{BD}+
\delta_{AD}\eta^{{a}}_{BC}\,, \label{hooft4}
\end{eqnarray}
where summations over repeated indices are assumed. The analogical identities are valid for $\bar\eta^{{a}}_{AB}\,$,
except  for (\ref{hooft2}), (\ref{hooft4}), which are modified for $\bar\eta^{{a}}_{AB}$ as
\begin{eqnarray}\label{bhooft-2}
\bar\eta^{{a}}_{AB}\bar\eta^{{a}}_{CD}&=&\delta_{AC}\delta_{BD}-\delta_{AD}\delta_{BC}-\epsilon_{ABCD}\,,
\\[4pt]
\label{bhooft-4}
\epsilon_{BCDE}\bar\eta^{{a}}_{AE}&=& -\delta_{AB}\bar\eta^{{a}}_{CD}+\delta_{AC}\bar\eta^{{a}}_{BD}-
\delta_{AD}\bar\eta^{{a}}_{BC}\,.
\end{eqnarray}

Some corollaries of the relations (\ref{hooft4}), (\ref{bhooft-4}) are
\begin{equation}\label{ide}
\epsilon_{ABCD}\eta^{{a}}_{CD}=2\,\eta^{{a}}_{AB}\,,\qquad
\epsilon_{ABCD}\bar\eta^{{a}}_{CD}=-2\,\bar\eta^{{a}}_{AB}\,.
\end{equation}

Finally, we note the important relations between $\eta$ and $\bar\eta$:
 \be
\label{ide-imp}
 \eta^{{a}}_{C[A}\bar\eta^{{b}}_{B]C}=0\,, \qquad
\eta^a_{AB} \bar \eta^b_{AB} = 0 \,.
 \ee

\renewcommand\theequation{B.\arabic{equation}} \setcounter{equation}0
\section*{Appendix B. Mirror multiplet: alternative conventions}

Let us change the order of indices and define the superfield $Y^{\alpha' i'}(t, \theta)$ subject to the constraints
\begin{equation}\label{constnew}
D^{k({i}^\prime}Y^{\alpha^\prime {j}^\prime)  } = 0\,,
\end{equation}
The solution to these constraints is
\begin{equation}
\label{Ynew}
Y^{\alpha^\prime  {i}^\prime  } = y^{\alpha^\prime {i}^\prime  }  -   \psi_{k}^{\alpha^\prime} \theta^{k{i}^\prime} -
i \dot y^{\alpha^\prime  j^\prime  }
\theta_{{j}^\prime k } \theta^{k{i}^\prime} +
{\textstyle\frac{i}{3}}\,\dot\psi_{j}^{\alpha^\prime} \theta^{j{k}^\prime}\theta_{{k}^\prime k}\theta^{k{i}^\prime}
+ {\textstyle\frac{1}{12}}\,\ddot y^{\alpha^\prime {i}^\prime  } \,
\theta^{k{k}^\prime} \theta_{{k}^\prime j} \theta^{j{j}^\prime} \theta_{{j}^\prime k}\,.
\end{equation}
This differs from \p{Y} by a more smooth flow of indices, which facilitates  the conversion to vector notations. To this end,
it is natural to stick, instead of
\p{4vect-2} -- \p{4vect-2a}, to the definitions

\begin{equation}\label{ysigma}
y_{{i}^\prime \alpha^\prime }={\textstyle\frac{1}{\sqrt{2}}}\,y^{M}(\sigma_M)_{{i}^\prime \alpha^\prime }\,,\quad
y^{\alpha^\prime {i}^\prime}=-{\textstyle\frac{1}{\sqrt{2}}}\,y^{M}(\sigma^\dagger_M)^{\alpha^\prime {i}^\prime}\,,
\end{equation}
\begin{equation}\label{pYpart}
Y^{\alpha^\prime {i}^\prime}=-{\textstyle\frac{1}{\sqrt{2}}}\,Y^{M}(\sigma^\dagger_M)^{\alpha^\prime {i}^\prime}\,,
\quad
\partial^{\alpha^\prime {i}^\prime }={\textstyle\frac{1}{\sqrt{2}}}\,\partial_{M}(\sigma^{M\dagger})^{\alpha^\prime {i}^\prime }\,,
\quad
\psi^{\alpha^\prime i}=-{\textstyle\frac{1}{\sqrt{2}}}\,Y^{M}(\sigma^\dagger_M)^{\alpha^\prime i}\,,
\end{equation}
 etc.

Then the vector superfield $Y^M$ is expressed in the form analogous to \p{XA}, with $\eta$ replaced by $\bar \eta$,
   \be
 \lb{YM}
 Y^M \ =\ y^M - \frac i2 \bar \eta^a_{NP} \theta^N \theta^P \bar \eta^a_{MQ} \dot y^Q -
 \frac 1{24} \epsilon_{NPQS} \theta^N \theta^P \theta^Q \theta^S \, \ddot y^M \ + \ {\rm fermion\ terms} \, .
  \ee
The constraints \p{constnew} imply, in particular, $\bar \eta^a_{MN} D_M Y^N = 0\,$.

By the same token as in \p{susy-char1-s-2a}, \p{susy-char1-tr-2a}, one can define the HKT supercharges with the complex structures
\begin{equation}\label{Inew}
({I}^{{a}})_{M}{}^{N}= \bar\eta^{{a}}_{MN}\,.
\end{equation}
 In fact, the structures \p{Inew} are obtained from the structures \p{compl-str-2} by the conjugation, $I \to \Omega I \Omega\,$, with $\Omega =
 {\rm diag}(1,1,1,-1)$. This corresponds to changing the sign of $x_4$.

For some purposes, the representation \p{Ynew} is more convenient than \p{Y}. For example, the proof of the absence of
mixed kinetic terms in the component Lagrangian for two interacting mutually mirror multiplets is obtained at almost no price.
Let us use the representations
\p{XA} and \p{YM} for $X^A$ and $Y^M$. Mixed bosonic kinetic terms in \p{act} could be generated by the structure
 \be
\lb{nomixed1}
 \sim \frac {\partial L}{\partial X^A \partial Y^M } \eta^a_{BC} \theta^B \theta^C  \eta^a_{AD} \dot x^D
 \bar \eta^b_{NP} \theta^N \theta^P  \bar \eta^b_{MQ} \dot y^Q \, ,
 \ee
which is, however, vanishing due to \p{ide} and \p{ide-imp}. Another potential source,
  \be
\lb{nomixed2}
 \sim \frac {\partial L}{\partial X^A} \theta^4 \ddot x^A -  \frac {\partial L}{\partial Y^M} \theta^4 \ddot y^M\,,
 \ee
also contributes zero due to the opposite signs of the last terms in \p{XA} and \p{YM}.

New conventions for the mirror multiplet imply the following form of the 8-dimensional CKT complex structures,
    \begin{equation}\label{structnew}
{I}^{{a}}=
\left(
\begin{array}{cc}
- \eta^{{a}}_{AB} & 0 \\
0 &  \bar\eta^{{a}}_{MN} \\
\end{array}
\right).
\end{equation}

A disadvantage of such conventions is that the OKT structure in a flat 8-dimensional space, or in a conformally flat 8-dimensional
manifold with harmonic conformal factor, becomes less explicit. The four extra supercharges can  still be expressed in the form
\p{h-susy-char1-s-2g}, \p{h-susy-char1-tr-2g} with the complex structures

\begin{equation}\label{dopstruct}
\tilde I=
\left(
\begin{array}{cc}
0 & -\Omega \\
\Omega & 0 \\
\end{array}
\right)
, \qquad
\tilde I^{{p}}=
\left(
\begin{array}{cc}
0 & \bar\eta^p \Omega \\
\Omega \bar\eta^p & 0 \\
\end{array}
\right),
\end{equation}
where $\Omega$ was defined above. They are obtained from \p{ad-compl-str-free} by conjugation with the matrix $\sim \mbox{diag} (1, \Omega)$.
The set \p{structnew} and \p{dopstruct} satisfies the $D=8$ Clifford algebra.
The representation \p{ad-compl-str-free} looks more elegant compared to \p{dopstruct}, and this determined
our choice of the conventions in the main text.

\end{document}